\documentclass[a4paper,twocolumn,11pt,unpublished]{quantumarticle}
\pdfoutput=1
\usepackage[utf8]{inputenc}
\usepackage[english]{babel}
\usepackage[table,dvipsnames]{xcolor}
\usepackage[T1]{fontenc}
\usepackage{amsmath}
\usepackage{subfigure}
\usepackage{amssymb}
\usepackage{tikz}
\usepackage{lipsum}
\usepackage{graphicx}
\usepackage{epsfig}
\usepackage{epstopdf}
\usepackage{braket}
\usepackage{comment}
\usepackage{hyperref}
\hypersetup{
    colorlinks=true,
    linkcolor=teal,
    citecolor=teal,
    urlcolor=teal,
    breaklinks=true,    
    hypertexnames=false
}
\usepackage{amsthm}
\usepackage{thm-restate}
\usepackage{apptools}
\usepackage{enumitem}
\usepackage{dsfont}
\usepackage{bbold}
\usepackage{bm}
\usepackage{eurosym}
\usepackage{diagbox}
\usepackage[numbers,sort&compress]{natbib}
\usepackage[ruled,vlined]{algorithm2e}
\usepackage{tikz}
\usetikzlibrary{shapes, arrows, positioning}
\tikzstyle{process} = [rectangle, rounded corners, minimum width=5cm, minimum height=1cm, text centered, draw=black, fill=white, align=center, font = \fontsize{9}{11}\selectfont]
\tikzstyle{arrow} = [thick,->,>=stealth]

\newcommand{\Input}{\State \textbf{Input:} }
\newcommand{\Output}{\State \textbf{Output:} }

\newcommand{\be}{\begin{equation}}
\newcommand{\ee}{\end{equation}}
\newcommand{\ba}{\begin{aligned}}
\newcommand{\ea}{\end{aligned}}

\newcommand{\R}{\mathbb{R}}
\newcommand{\bc}{\begin{center}}
\newcommand{\ec}{\end{center}}
\newcommand{\beq}{\begin{equation}}
\newcommand{\eeq}{\end{equation}}
\newcommand{\beqq}{\begin{equation*}}
\newcommand{\eeqq}{\end{equation*}}
\newcommand{\beqa}{\begin{align}}
\newcommand{\eeqa}{\end{align}}
\newcommand{\barr}{\begin{array}}
\newcommand{\earr}{\end{array}}
\newcommand{\bi}{\begin{itemize}}
\newcommand{\ei}{\end{itemize}}

\newcommand{\C}{\mathbb{C}}

\newcommand{\eexp}{\mathrm{eexp}}

\newtheorem{lem}{Lemma}
\newtheorem{theo}{Theorem}

\newtheorem{defi}{Definition}

\newtheorem{algo}{Algorithm}
\AtAppendix{\counterwithin{lem}{section}}
\AtAppendix{\counterwithin{defi}{section}}
\AtAppendix{\counterwithin{theo}{section}}
\DeclareMathOperator{\poly}{poly}

\DeclareMathOperator{\Tr}{Tr}

\DeclareMathOperator{\N}{\mathbb{N}}
\DeclareMathOperator{\I}{\mathbb{I}}

\newcommand{\PSPACE}{\mathsf{PSPACE}}
\newcommand{\NCpoly}{\mathsf{NC}(\poly)}
\newcommand{\PTIME}{\mathsf{PTIME}}
\newcommand{\DSPACE}{\mathsf{DSPACE}}
\newcommand{\NC}{\mathsf{NC}}

\newcommand{\CC}{\mathbb{C}}
\newcommand{\QQ}{\mathbb{Q}}

\newcommand{\calO}{\mathcal{O}}

\begin{document}

\title{Bounding the computational power of bosonic systems}

\author{Varun Upreti}
\email{varun.upreti@inria.fr}
\affiliation{DIENS, \'Ecole Normale Sup\'erieure, PSL University, CNRS, INRIA, 45 rue d’Ulm, Paris, 75005, France}
\author{Dorian Rudolph}
\affiliation{Department of Computer Science and Institute for Photonic Quantum Systems (PhoQS), Paderborn University, Warburger Str. 100, 33098 Paderborn, Germany}
\author{Ulysse Chabaud}
\affiliation{DIENS, \'Ecole Normale Sup\'erieure, PSL University, CNRS, INRIA, 45 rue d’Ulm, Paris, 75005, France}

\begin{abstract}
Bosonic quantum systems operate in an infinite-dimensional Hilbert space, unlike discrete-variable quantum systems. This distinct mathematical structure leads to fundamental differences in quantum information processing, such as an exponentially greater complexity of state tomography \cite{AnnaMele2024} or a factoring algorithm in constant space \cite{brenner2024factoring}. Yet, it remains unclear whether this structural difference of bosonic systems may also translate to a practical computational advantage over finite-dimensional quantum computers. Here we take a step towards answering this question by showing that universal bosonic quantum computations can be simulated in polynomial space (and exponential time) on a classical computer, significantly improving the best previous upper bound requiring exponential memory \cite{chabaud2024complexity}. In complexity-theoretic terms, we improve the best upper bound on \textsf{CVBQP} from \textsf{EXPSPACE} to \textsf{PSPACE}. This result is achieved using a simulation strategy based on finite energy cutoffs and approximate coherent state decompositions. While we propose ways to potentially refine this bound, we also present arguments supporting the plausibility of an exponential computational advantage of bosonic quantum computers over their discrete-variable counterparts. Furthermore, we emphasize the role of circuit energy as a resource and discuss why it may act as the fundamental bottleneck in realizing this advantage in practical implementations.   
\end{abstract}
\maketitle
\section{Introduction}
Quantum advantage has been a key research goal for decades \cite{shor1994algorithms,Aaronson2013,Ronnow2014,Harrow2017,Boixo2018}, with a major focus on identifying the best physical platform for quantum computing \cite{kok2007linear,pogorelov2021compact,xue2021cmos,acharya2024quantum,LogicalLukin}. These platforms differ not only in their physical properties but also in the type of quantum computations they may perform. While systems like neutral atoms and quantum dots operate in finite-dimensional Hilbert spaces, performing discrete-variable (DV) computations over qubits or qudits, bosonic platforms such as photonics, vibrational degrees of freedom of trapped-ion systems, or superconducting cavity modes, are described by infinite-dimensional Hilbert spaces, and may perform continuous-variable (CV) computations over qumodes or simply modes, dealing with operators taking on a continuous spectrum of values, like position and momentum.

Bosonic platforms have emerged as promising candidates for quantum computing \cite{Gottesman2001,Knill2001,menicucci2006universal,madsen2022quantum}, driven by experimental breakthroughs such as the generation of large-scale entangled states \cite{yokoyama2013ultra} and remarkable error correction capabilities \cite{sivak2023}. With recent results such as the exponential complexity of learning the state of CV quantum systems \cite{AnnaMele2024} and a factoring algorithm using a constant number of bosonic modes \cite{brenner2024factoring} suggesting that the unique mathematical structure of bosonic systems leads to novel properties, this raises a compelling question: \textit{can the infinite-dimensional Hilbert space inherent to these platforms lead to a computational advantage over DV quantum computers?}

The above results support a positive answer to the question. This would refute the quantum-extended Church--Turing thesis \cite{kaye2007introduction}, which asserts that any efficient physical computation can be efficiently simulated by a DV quantum computer, not only reshaping our understanding of quantum computation but also establishing bosonic platforms as foundational to future quantum technologies. 
On the other hand, other results suggest that the CV and DV paradigms may not differ fundamentally in computational power \cite{arzani2025}. Indeed, one reason why bosonic systems might challenge the quantum-extended Church--Turing thesis is their theoretically unbounded energy, but any practical bosonic quantum computer operates with finite energy. That being said, a negative result to the question would also be highly valuable, as it could lead to the development of quantum algorithms in finite-dimensional Hilbert spaces for simulating bosonic dynamics.

\begin{figure}
    \centering
    \includegraphics[width=\linewidth]{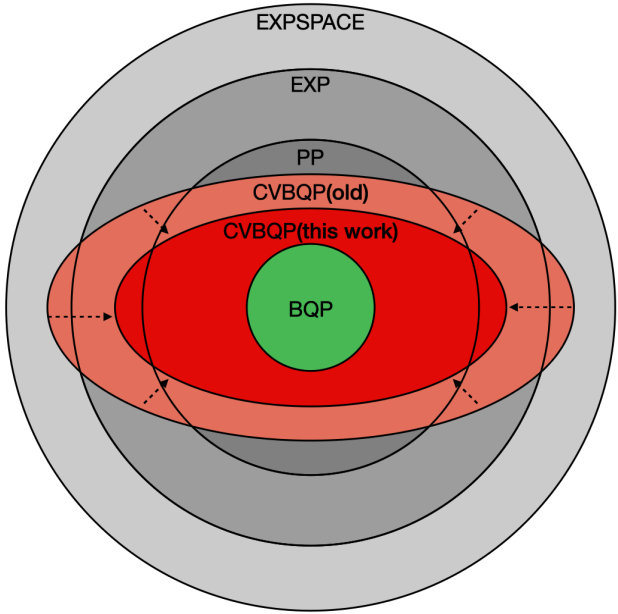}
    \caption{Relationship between several classical and quantum complexity classes (see \hyperref[sec:methods]{Methods}). We improve the upper bound on \textsf{CVBQP} from \textsf{EXPSPACE} \cite{chabaud2024complexity} to \textsf{PSPACE}.}
    \label{fig:complexity_class}
\end{figure}
To characterize the computational power of CV and DV quantum systems, we need to examine the type of problems each can solve efficiently. This is achieved using complexity theory, the branch of theoretical computer science that studies the resources---such as time, space, and energy---required to solve computational problems, aiming to classify problems into complexity classes based on their inherent difficulty \cite{papadimitriou1994computational,Aaronson2013book,chabaud2024complexity}. Quantum complexity theory for DV systems is well-established: for instance, the class of problems solvable by a DV quantum computer in polynomial time with high probability, termed \textsf{BQP}, is a subset of the complexity class \textsf{PP} \cite{Adleman1997}, where $\textsf{PP}$ encompasses problems solvable by a probabilistic classical computer in polynomial time with a success rate greater than 50\%. Note that there is likely a large separation between the two classes, but proving unconditional separation between computational complexity classes is often a daunting task.

On the other hand, complexity theory for bosonic systems is still in its early stages, with the main challenge being the definition of a reasonable complexity class for problems solvable by a bosonic quantum computer---one that is neither overly optimistic nor too restrictive (for example, by simply restricting bosonic computations to DV computations). In a recent attempt to establish the foundations for a bosonic quantum computational complexity theory in \cite{chabaud2024complexity}, the authors introduce the complexity class \textsf{CVBQP}, which represents problems solvable by bosonic circuits using a standard gate set \cite{lloyd1999quantum,Sefi2011} in polynomial time\footnote{The authors denote this class as $\textsf{CVBQP}[X^3]$, but for brevity, we refer to it here as $\textsf{CVBQP}$.}. In particular, they show the inclusion $\textsf{CVBQP} \subseteq \textsf{EXPSPACE}$, where \textsf{EXPSPACE} encompasses problems solvable by a classical computer using exponential memory, without strict time limitations. \textsf{EXPSPACE} strictly includes \textsf{PP} and thus \textsf{BQP}, and contains problems so complex that they demand an astronomical amount of memory, far beyond the capacity of any physical computer.

Here we improve on this result by showing that continuous-variable quantum computations can be simulated in polynomial space on a classical computer, namely $\textsf{CVBQP}\subseteq \textsf{PSPACE}$. This result brings the computational power of bosonic systems significantly closer to that of DV quantum systems (see Figure \ref{fig:complexity_class}), thus limiting the potential computational advantage of CV quantum systems over their DV counterparts.

The rest of the paper is organized as follows. Section \ref{sec:results} demonstrates that bosonic computations can be simulated in polynomial space, presenting a step-by-step derivation along with discussions of the physical significance of each step. We conclude this section by highlighting the key techniques that allow us to establish a stricter upper bound for \textsf{CVBQP} than previous work, as well as the barriers to improving this bound. Section \ref{sec:conclusion} discusses the broader implications of our findings, highlights proof techniques that may be of independent interest, and outlines directions for future research. Finally, Section \ref{sec:methods} introduces the additional technical tools and definitions required for our proofs.

\section{Results}\label{sec:results}

This section presents our main result, Theorem \ref{theo:main_result}, along with the intermediate results that contribute to its proof (see section \ref{sec:main}). We discuss their physical significance and provide a step-by-step explanation of how they fit together. An outline of the proof is given in Figure \ref{fig:proof_sketch}. We discuss potential limitations for using our proof techniques to improve the simulation of bosonic computations in section \ref{sec:PSPACE}. We first provide a few preliminary notations and definitions in section \ref{sec:preli} hereafter.

\subsection{Preliminaries}
\label{sec:preli}

We refer the reader to \cite{NielsenChuang} for background on quantum information theory and to \cite{Braunstein2005,ferraro2005gaussian,Weedbrook2012} for CV quantum information material. Hereafter, the sets $\N, \R$ and $\C$ are the set of natural, real, and complex numbers respectively, with a * exponent when $0$ is removed from the set.

In CV quantum information, a mode refers to a degree of freedom associated with a specific quantum field of a CV quantum system, such as a single spatial or frequency mode of light, and is the equivalent of a qubit in the CV regime. In this paper, $m \in \N^*$ denotes the number of modes in the system. We are concerned with the scaling of the complexity of simulation with respect to input number of modes $m$, and we use the following scaling notations for brevity: $\poly\equiv\mathcal \calO(\poly(m))$, $\exp\equiv\mathcal \calO(\exp(\poly(m)))$, and $\eexp\equiv\mathcal \calO(\exp(\exp(\poly(m))))$, whereas $1/\poly \equiv \mathcal{O}(1/\poly)$, and similarly for $1/\exp$ and $1/\eexp$.

Hereafter, $\{\ket{n} \}_{n\in\N}$ denote the single-mode Fock basis, with $\ket0$ being the vacuum state, and $\hat{a}$ and $\hat{a}^\dagger$ refer to the single-mode annihilation and creation operator, respectively, satisfying $[\hat{a},\hat{a}^\dagger] = \I$. These are related to the position and momentum quadrature operators as
\begin{equation}
    \hat{q} = \frac{1}{\sqrt{2}}(\hat{a} + \hat{a}^\dagger), \hspace{5mm}\hat{p} = \frac{-i}{\sqrt{2}}(\hat{a} - \hat{a}^\dagger),
\end{equation} 
with the convention $\hbar = 1$. Furthermore, $\hat{q}$ and $\hat{p}$ satisfy the commutation relation $[\hat{q},\hat{p}] = i\I$. The particle number operator is given as $\hat{n} = \hat{a}^\dagger \hat{a}$.

Products of unitary operations generated by Hamiltonians that are quadratic in the quadrature operators of the modes are called Gaussian unitary operations, and states produced by applying a Gaussian unitary operation to the vacuum state are Gaussian states. The action of an $m$-mode Gaussian unitary operation $\hat{G}$ on the vector of quadratures $\boldsymbol{\Gamma}= [\hat{q}_1,\dots,\hat{q}_m,\hat{p}_1,\dots,\hat{p}_m]$ is given by
\begin{equation}\label{eqn:symplectic_transform_quadrature}
    \hat{G}^\dagger \boldsymbol{\Gamma} \hat{G} = S \boldsymbol{\Gamma} + \boldsymbol{d},
\end{equation}
where $S$ is a $2m \times 2m$ symplectic matrix and $\boldsymbol{d} \in \R^{2m}$ is a displacement vector.

Passive linear unitary gates, such as beam splitters and phase shifters, are Gaussian transformations that preserve the total particle number. Single-mode displacement operators and single-mode squeezing operators are defined as $\hat{D}(\alpha) = e^{\alpha \hat{a}^\dagger - \alpha^* \hat{a}}$ and $\hat{S} (\xi) = e^{\frac12 (\xi \hat{a}^2 - \xi^* \hat{a}^{\dagger 2})}$ respectively, where $\alpha, \xi \in \C$. A coherent state $\ket{\alpha}$ is generated by the action of displacement operator on the vacuum state $\ket{\alpha} = \hat{D}(\alpha)\ket{0}$, and is the eigenstate of the annihilation operator, i.e., $\hat{a}\ket{\alpha} = \alpha \ket{\alpha}$, $\forall \alpha \in \C$. In what follows, displaced passive linear unitary gates refer to passive linear unitary gates multiplied with a tensor product of single-mode displacement operators. A crucial property of (displaced) passive linear unitary gates which we will use is that they transform tensor products of coherent states to tensor products of coherent states. The Bloch--Messiah decomposition \cite{ferraro2005gaussian} states that any multimode Gaussian unitary gate can be decomposed into a product of orthogonal gates, rotation (or phase-shift) gates, displacement gates and shearing gates.

Non-Gaussian operations are necessary for enabling quantum advantage since Gaussian gates acting on Gaussian states and followed by Gaussian measurements can be classically simulated efficiently \cite{Bartlett2002}. One prominent example of a non-Gaussian gate is the cubic phase gate $e^{i\gamma\hat{q}^3}$ \cite{Gottesman2001}, whose action on the quadratures is given by \cite{Budinger2024}
\begin{equation}
    e^{-i\gamma\hat{q}^3} \hat{q} e^{i\gamma\hat{q}^3}   = \hat{q}, \hspace{5mm} e^{-i\gamma\hat{q}^3}  \hat{p}  e^{i\gamma\hat{q}^3}    = \hat{p} + 3\gamma \hat{q}^2.
\end{equation}
A standard model of CV quantum computation (CVQC) is defined by a vacuum state input into a circuit with Gaussian unitaries and cubic phase gates \cite{lloyd1999quantum,Sefi2011,arzani2025}.

Complexity classes (see \hyperref[sec:methods]{Methods} and \cite{papadimitriou1994computational,Aaronson2013book,chabaud2024complexity}) categorize the difficulty of problems based on their resource requirements, providing a framework for comparing the computational power of DV and CV systems by analyzing the complexity of the problems they can solve.
A complete problem for a complexity class is a problem that is both a member of the class and at least as hard as every other problem within the class. Specifically, any problem in the class can be reduced to it in polynomial time (or according to the appropriate reduction for the class). Consequently, if one can solve a complete problem in a class, one can solve any other problem within that class as well.
The complexity class \textsf{CVBQP} comprises problems that can be solved in polynomial time by a bosonic quantum computer using Gaussian and cubic phase gates on an input vacuum state, together with particle number measurements. The complete problem for the \textsf{CVBQP} complexity class can be stated as follows: Given two predefined boson number regions, determine with high probability in which of these regions the number of bosons in the first mode of a bosonic quantum computer—equipped with a polynomial number of Gaussian and cubic phase gates—lies with high probability, assuming it belongs to one of the two regions.

Finally, in this work we use three different notions of closeness of quantum states in the paper: trace distance ($D$), fidelity ($F$) and $2$-norm distance. Definitions and relationships between them are given in \hyperref[sec:methods]{Methods}.

\subsection{Main result}
\label{sec:main}

With the necessary background established, we now present the main result of the paper, which substantially refines the computational complexity of simulating bosonic computations with Gaussian unitaries and cubic phase gates, bringing it closer to that of discrete-variable quantum systems. This result is formalized as follows:

\begin{theo}\label{theo:main_result} Continuous-variable quantum computations can be simulated in polynomial space on a classical computer. Formally: 
\begin{equation}
    \mathsf{CVBQP}\subseteq \mathsf{PSPACE}.
\end{equation}
\end{theo}

\noindent Here recall that \textsf{PSPACE} is the class of problems solvable by a classical computer in exponential time using polynomial memory and $\textsf{CVBQP}$ represents problems solvable by a bosonic quantum computer with input vacuum, Gaussian unitary gates and cubic phase gates, and particle number measurements in polynomial time. 

Theorem \ref{theo:main_result} shows that problems solvable by bosonic quantum computers will at most be exponentially hard as compared to the ones solvable by DV quantum computers. This draws a parallel to the existing result that the tomography of bosonic systems is exponentially hard as compared to the tomography of DV systems \cite{AnnaMele2024}. While this result already significantly narrows the potential gap between the computational power of DV and CV quantum systems, it does not rule out the possibility of a tighter bound. However, with our current methods, it remains unclear how such an improvement could be achieved. We elaborate the last point at the end of the section.

In what follows, we give a step-by-step overview of the proof of Theorem \ref{theo:main_result} (see Figure \ref{fig:proof_sketch}).

\subsection{A complete problem for \textsf{CVBQP}}

Firstly, we define a simple problem of probability estimation at the output of a specific class of bosonic circuits:
\begin{defi}[\textsc{PassiveBosonNumEst}]\label{defi:prob_passive}
     Given an $m$ mode vacuum state $\ket{0}^{\otimes m}$ evolving through a unitary circuit
   \begin{equation}\label{eqn:U_t_analysis}
    \hat{U}_t = \hat{G} e^{i\gamma_t\hat{q}_1^3} \hat{V}_{t-1} \dots \hat{V}_1 e^{i\gamma_1\hat{q}_1^3}\hat{V}_0,
\end{equation}
where $\hat{V}_0,\dots,\hat{V}_{t-1}$ are displaced passive linear unitary gates, $\hat{G}$ is a Gaussian unitary gate, $t = \poly$, and $\gamma$ and the elements of the symplectic matrix and displacement vector corresponding to $\hat{G}$ and $\hat{V}_i$ $\forall i$ are specified upto $m$ bits of precision, i.e. their magnitude is $\leq \exp$, give the probability of number measurement on the first mode
\begin{equation}\label{eqn:prob_to_calculate}
    P(n) = \Tr[\ket{n}\!\bra{n}\otimes \I_{m-1} \hat{U}_t \ket{0}\! \bra{0}^{\otimes m} \hat{U}_t^\dagger]
\end{equation}
up to additive precision $1/\poly$ for $n = \poly$.
\end{defi}

\noindent We show that this problem in fact captures the complexity of $\textsf{CVBQP}$ computations:

\begin{restatable}{theo}{thmCVBQPcomplete}\label{theo:CVBQP_complete_desc}
    \hyperref[defi:prob_passive]{\textsc{PassiveBosonNumEst}} is $\mathsf{CVBQP}$-complete.
\end{restatable}
\noindent The proof of Theorem \ref{theo:CVBQP_complete_desc} is given in Appendix \ref{appendixsec:CVBQP_complete_desc}, and involves mapping the definition of the $\textsf{CVBQP}$-complete problem given in \cite[Definition 4.5]{chabaud2024complexity} in terms of estimating single number measurement probability in a bosonic circuit with Gaussian unitary gates and cubic phase gates up to the given precision. Then, Lemma \ref{lem:passive_reduction} below allows us to replace the generic Gaussian unitary gates in the circuit with displaced passive linear unitary gates, leaving only a final Gaussian unitary gate, and allowing us to describe the circuit of the \textsf{CVBQP}-complete problem by the unitary $\hat{U}_t$ given by Eq.~(\ref{eqn:U_t_analysis}), thus completing the proof.
\begin{restatable}[Gaussian and cubic circuit reduction]{lem}{lemPassiveReduction}\label{lem:passive_reduction}
    Given a single-mode cubic gate $e^{i\gamma \hat{q}_1^3}$ and an $m$-mode Gaussian unitary $\hat{G}$, there exist an $m$-mode displaced passive linear unitary $\hat B$ and a cubicity parameter $\gamma'\in\mathbb R$ such that
    \begin{equation}
        \hat{G}^\dagger e^{i\gamma \hat{q}_1^3} \hat{G} = \hat{B}^\dagger e^{i\gamma' \hat{q}_1^3} \hat{B}.
    \end{equation}
    Moreover, $\hat{B}$ and $\gamma'$ can be computed efficiently given the descriptions of $\hat{G}$ and $\gamma$.
\end{restatable}
\noindent The proof of Lemma \ref{lem:passive_reduction} is given in Appendix \ref{appendixsec:proof_lem_passive_reduction} and is based on the observation that the symplectic action of $G$ on the position quadrature $\hat{q}_1$ (in $e^{i\gamma\hat{q}_1^3}$) can be mapped to a rescaling of $\hat{q}_1$ by adding displaced passive linear unitaries on both sides. The scaling constant is then absorbed into the cubicity.

Note that Lemma \ref{lem:passive_reduction} is of independent interest, as it shows that most of the squeezing present in Gaussian unitary gates (which can be identified through their Bloch--Messiah decomposition) within a circuit of Gaussian unitaries and cubic phase gates can be eliminated without increasing the number of cubic phase gates. This is achieved by transforming a Gaussian unitary gate (which may include squeezing) into a displaced passive linear unitary (without squeezing) by modifying the cubicity, with the only cost being the addition of a final Gaussian unitary at the end. This result improves upon \cite{Sefi2011}, where squeezing in an $m$-mode Gaussian unitary gate is removed by first decomposing it into a passive linear unitary, displacement gates, and $m$ squeezing gates via Bloch-Messiah decomposition \cite{ferraro2005gaussian}. Each squeezing gate is then expressed as a combination of a rotation and a shearing gate \cite{chabaud2021holomorphic}, with the shearing gate further decomposed into displacement gates and two cubic phase gates \cite{Sefi2011}. This process requires adding $2m$ extra cubic phase gates per Gaussian unitary, totaling $2m+t$ for $t$ such gates. In contrast, Lemma \ref{lem:passive_reduction} eliminates squeezing while keeping the number of cubic phase gates fixed at $t$.

\subsection{A classical algorithm for \hyperref[defi:prob_passive]{\textsc{PassiveBosonNumEst}}}

The following theorem gives an upper bound on the computational complexity of \hyperref[defi:prob_passive]{\textsc{PassiveBosonNumEst}}:
 \begin{theo}\label{theo:complexity_prob}
     \hyperref[defi:prob_passive]{\textsc{PassiveBosonNumEst}} $\in\mathsf{PSPACE}$.
 \end{theo}

\noindent Together with Theorem \ref{theo:CVBQP_complete_desc}, this result shows that $\mathsf{CVBQP}\subseteq \mathsf{PSPACE}$, thus completing the proof of Theorem \ref{theo:main_result}.

The rest of this section is devoted to the proof of Theorem \ref{theo:complexity_prob}. In section \ref{sec:cohdec}, we introduce coherent state decompositions (Lemma \ref{lem:csr_decomp_cubic}), which we use in section \ref{sec:algo} to obtain an algorithm for computing a simple approximation of the output state of the circuit in \hyperref[defi:prob_passive]{\textsc{PassiveBosonNumEst}} as a finite superposition of Gaussian states (Algorithm \ref{algo:decomp}). We further characterize the computational complexity of this algorithm (Lemma \ref{lem:csr}). Finally, we explain how to computer the number statistics of the approximated state in section \ref{sec:statsG} (Lemma \ref{lem:prob}) to conclude the proof.

\subsubsection{Coherent state decompositions}
\label{sec:cohdec}

The central step in proving Theorem \ref{theo:complexity_prob} is to approximate the output state $\ket{\psi_\mathrm{out}}=\hat{U}_t \ket{0}^{\otimes m}$ (see Eq.~(\ref{eqn:U_t_analysis})) as a superposition of Gaussian states. We achieve this using successive coherent state decompositions. More precisely, we approximate the action of a cubic phase gate on a coherent state $e^{i\gamma\hat{q}^3}\ket{\alpha}$ by a superposition of coherent states. This amounts to bounding the $\epsilon$-approximate coherent state rank, which is a measure of non-classicality of a quantum state and is defined as follows:

\begin{defi}[$\epsilon$-approximate coherent state rank \cite{Gehrke2012,sperling2015convex}] The approximate coherent state rank of a pure state $\ket{\psi}$ is the minimal $n$ such that there exists a normalized superposition of $n$ coherent states $\ket{\psi_\alpha}$ such that
\begin{equation}
    D(\psi,\psi_\alpha) \leq \epsilon.
\end{equation}
\end{defi}

\noindent Our next result formalizes how this approximation can be done to the required precision:

\begin{restatable}[Coherent superposition approximation of $e^{i\gamma\hat{q}^3}\ket{\alpha}$]{lem}{lemcsrdecompcubic}\label{lem:csr_decomp_cubic}
     For $\alpha \in \C, \gamma \in \R$ with $\alpha,\gamma = \exp$, let $\ket{\psi_{\alpha,\gamma}} := e^{i\gamma\hat{q}^3}\ket{\alpha}$. For $\epsilon = 1/\exp$, the $\epsilon$-approximate coherent state rank of $\ket{\psi_{\alpha,\gamma}}$ is upper bounded by $N = \exp$. The description of a corresponding coherent state approximation $\ket{\tilde{\psi}_{\alpha,\gamma}}$ can be computed in parallel polynomial time and is such that the $2$-norm distance $\left\| \ket{\psi_{\alpha,\gamma}} -  \ket{\tilde{\psi}_{\alpha,\gamma}}\right\| \leq 1/\exp$, where $\|\cdot\|$ is the $2$-norm in Hilbert space.
\end{restatable}

\noindent The proof of Lemma \ref{lem:csr_decomp_cubic} is given in Appendix \ref{appendixsec:proof_csr_decomp_cubic}. It involves using the gentle measurement lemma \cite{Wilde2013} to truncate $\ket{\psi_{\alpha,\gamma}}$ to a superposition of Fock states of finite size, depending on the energy of $\ket{\psi_{\alpha,\gamma}}$ and the desired precision. Then, applying Theorem 1 from \cite{Marshall2023} allows us to find a coherent state decomposition for the truncated Fock state. To ensure that the decomposition can be obtained in \emph{parallel polynomial time}, we give a method to compute the coefficients of the truncated state up to doubly exponential precision by generalizing techniques from \cite{chabaud2024complexity,Miatto2020fastoptimization} for approximating the Airy function.
In addition, we add a global phase to the approximated state to ensure the closeness in terms of $2$-norm distance rather than only in trace distance.
Here, \emph{parallel polynomial time} refers to computations that can be performed in polynomial time on a \emph{parallel random-access machine} (PRAM) \cite{FW78} (see \hyperref[sec:methods]{Methods}).
The sequential runtime would still be exponential, but we can show that the computations can be parallelized, even with doubly exponential precision, since arithmetic can be performed in polylogarithmic parallel time in the bit length.
Parallel polynomial time computations can be simulated using a deterministic Turing machine with polynomial space and exponential time (i.e. $\PSPACE$) \cite{FW78}.

In what follows, we explain how such coherent state decompositions may be used to approximate the output state $\ket{\psi_\mathrm{out}}=\hat{U}_t \ket{0}^{\otimes m}$ as a finite superposition of Gaussian states.

\subsubsection{Computing a Gaussian state decomposition}
\label{sec:algo}

Looking at the structure of the unitary circuit $\hat{U}_t$ in Eq.~(\ref{eqn:U_t_analysis}), when propagating the input $m$-mode vacuum state forward we only encounter cubic phase gates and displaced passive linear unitary gates, except for the Gaussian gate $\hat{G}$ at the very end of the circuit (see Figure \ref{fig:c_eql}). Since displaced passive linear unitary gates map tensor products of coherent states to tensor product of coherent states, and given that we start with an input coherent (vacuum) state, if we can approximate the action of a cubic phase gate on a coherent state $e^{i\gamma\hat{q}^3}\ket{\alpha}$ as a superposition of coherent states, carrying out the said approximation each time we encounter a cubic phase gate combined with the action of displaced passive linear unitary gates ensures that the output state before the final Gaussian gate $\hat{G}$ is well approximated by a finite superposition of coherent states. Applying $\hat G$ then yields a superposition of Gaussian states at the output. 
\begin{figure}
    \centering
    \includegraphics[width = 1.03\linewidth]{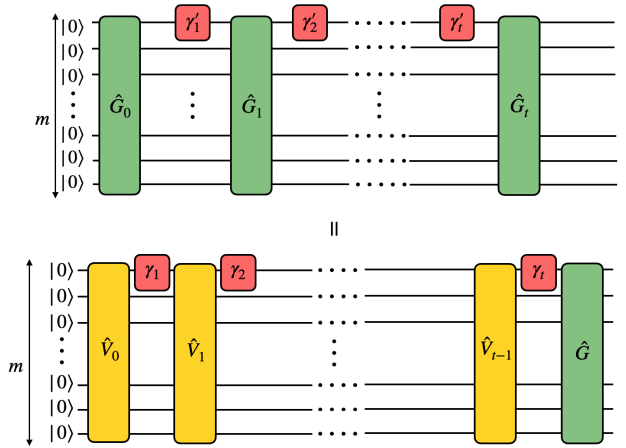}
    \caption{Bosonic circuit of the \textsf{CVBQP}-complete problem. $\hat{G},\hat{G}_0,\dots,\hat{G}_t$ are Gaussian unitary gates, $\hat{V}_0,\dots \hat{V}_{t-1}$ are displaced passive linear unitary gates, and $\gamma_1,\dots,\gamma_t,\gamma_1',\dots,\gamma_t'$ are cubic phase gates, with $t = \poly$, $|\gamma_i|,|\gamma_i'| \leq \exp$ $\forall i \in \{1,\dots,t\}$ and the elements of all the Gaussian unitary and displaced passive linear unitary gates are specified upto $m$ bits of precision, i.e., their magnitude is $\leq \exp$.}
    \label{fig:c_eql}
\end{figure}
This leads to Algorithm \ref{algo:decomp}.
\begin{algorithm}[ht] \textbf{Input:} Displaced passive linear unitary gates $\hat{V}_0, \dots, \hat{V}_{t-1}$, Gaussian unitary gate $\hat{G}$, parameters $\gamma_1, \dots, \gamma_t$.

Initialize $\ket{\psi_0} = \ket{0}^{\otimes m}$

Initialize $i = 1$.

\While{$i \leq t$}{
Compute $\ket{\psi_i'} = \hat{V}_{i-1} \ket{\psi_{i-1}}$ 

Approximate $\ket{\psi_i}= e^{i\gamma_i\hat{q}_1^3} \ket{\psi_i'}$ as superposition of coherent states via application of Lemma \ref{lem:csr_decomp_cubic} on each term in the input superposition.

$i \rightarrow i+1$.
} 

Compute $\ket{\tilde{\psi}_\mathrm{out}} = \hat{G}\ket{\psi_t}$.

\textbf{return} Superposition of Gaussian states $\ket{\tilde{\psi}_\mathrm{out}}$ approximating $\ket{\psi_\mathrm{out}} = \hat{U}_t \ket{0}^{\otimes m}$

\caption{Approximation by a superposition of Gaussian states} \label{algo:decomp} \end{algorithm}

\noindent The following result formalizes the precision and time complexity when using this algorithm:

\begin{restatable}[Gaussian superposition approximation]{lem}{lemcsr}\label{lem:csr}
     Given an $m$-mode state $\ket{\psi_\mathrm{out}} = \hat{U}_t \ket{0}^{\otimes m}$, with $\hat{U}_t$ given in Eq.~(\ref{eqn:U_t_analysis}) and $t = \poly$, we can obtain the description of a superposition of Gaussian states $\ket{\tilde{\psi}_\mathrm{out}} = \sum_{i=0}^{N} c_i \hat{G}\ket{\alpha^{(i)}_1\alpha^{(i)}_2\dots\alpha^{(i)}_m}$ in parallel polynomial time using Algorithm \ref{algo:decomp}, such that $\ket{\tilde{\psi}_\mathrm{out}}$ is $1/\exp$ close in trace distance to $\ket{\psi_\mathrm{out}} = \hat{U}_t \ket{0}^{\otimes m}$.
\end{restatable}

\noindent The proof of Lemma \ref{lem:csr} is given in Appendix \ref{appendixsec:theo_csr}. It combines repeated applications of Lemma \ref{lem:csr_decomp_cubic} and the triangle inequality to bound the trace distance between $\ket{\psi_\mathrm{out}}$ and $\ket{\tilde{\psi}_\mathrm{out}}$. Note that Lemma \ref{lem:csr} can also be interpreted as a result on the upper bound of the $1/\exp$-approximate Gaussian rank of $\ket{\psi_\mathrm{out}}$ \cite{hahn2024classicalsimulation,Dias2024}.

Note that a $1/\poly$-precise approximation is sufficient for our purpose, but the time complexity of Algorithm~\ref{algo:decomp} remains the same whether we use $\epsilon = 1/\exp$ or $\epsilon = 1/\poly$, since it is dominated by the number of terms in the superposition, which is $\exp$.


\subsubsection{Number statistics of Gaussian state decompositions}
\label{sec:statsG}

Given the approximation of the output state of the given bosonic circuit $\hat{U}_t\ket{0}^{\otimes m}$ as a superposition of Gaussian states $\ket{\tilde{\psi}_\mathrm{out}}$ , estimating the probability generated by boson number measurement on the first mode, $P(n)$, given by Eq.~(\ref{eqn:prob_to_calculate}) is the problem of estimating the probability of a single-mode number measurement on a superposition of Gaussian states. This can be done in a variety of ways, for instance by using techniques from Gaussian boson sampling \cite{hamilton2017gaussian,quesada2018gaussian} computations. Here, we choose to do this by approximating $\ket{n}$, the Fock state characterizing $P(n)$ with the measurement POVM $\ket{n}\bra{n} \otimes \I_{m-1}$, by a coherent state decomposition $\ket{\tilde{n}}$, using Theorem 1 of \cite{Marshall2023} to obtain 
\begin{equation}
        \ket{\tilde{n}} = \frac{1}{\sqrt{\mathcal{N}}} \sum_{j = 0}^n d_j \ket{\xi e^{2\pi i j/(n+1)}},
    \end{equation}
    where 
    \begin{equation}
    d_j = \frac{e^{\xi^2/2}}{n+1} \sqrt{n!} \frac{1}{\xi^n} e^{-2\pi i j/(n+1)}.
    \end{equation}
where the normalization factor $\mathcal{N} = 1 + \mathcal{O}(\xi^{2(n+1)}/(n+1)!)$. This state is $\chi$ close in trace distance to $\ket{n}$, with $\chi = \mathcal{O}\left(\xi^{(n+1)}/\sqrt{(n+1)!}\right)$, where $\chi$ and $\xi$ are free parameters. We take $\xi = 1/\poly$ that implies $\chi = 1/\exp$, and this state can be described in time $\poly$ (time to calculate the $n$ coherent state coefficients, with $n = \poly$). The choice of $\chi = 1/\exp$ or $1/\poly$ does not impact the time complexity, which depends on $n = \poly$.

Probability estimation with the descriptions of $\ket{\tilde{\psi}_\mathrm{out}}$ and $\ket{\tilde{n}}$ reduces to the problem of calculating Gaussian overlaps, where ideas from \cite{Dias2024} can be used. This is formalized in the following Lemma:

\begin{restatable}[Number statistics of superposition of Gaussian states]{lem}{lemprob}\label{lem:prob}
    Given an $m$-mode Gaussian state superposition $\ket{\tilde{\psi}_\mathrm{out}} = \sum_{i=0}^{N} c_i \ket{G_i}$, with $N=\exp$ and a single-mode coherent state superposition $\ket{\tilde{n}} = \sum_{i=0}^n d_i \ket{\beta_i}$, the probability
    \begin{eqnarray}
        \tilde{P}(n) &=& \Tr[\Pi_{\tilde{n}} \ket{\tilde{\psi}_\mathrm{out}}\bra{\tilde{\psi}_\mathrm{out}} ] \nonumber \\ &=& \bra{\tilde{\psi}_\mathrm{out}}(\ket{\tilde{n}}\bra{\tilde{n}}\otimes \I_{m-1})\ket{\tilde{\psi}_\mathrm{out}},
    \end{eqnarray}
    where $\Pi_{\tilde{n}} = \ket{\tilde{n}}\bra{\tilde{n}} \otimes \I_{m-1}$, can be computed in parallel polynomial time up to precision $1/\exp$.
\end{restatable}
\noindent The proof of Lemma \ref{lem:prob}, provided in Appendix \ref{appendixsec:lem_prob}, involves deriving the description of the $(m-1)$-mode Gaussian state resulting from the projection of a single-mode coherent state onto an $m$-mode Gaussian state. This is followed by calculating the overlaps of $n^2 \exp^2 = \exp$ Gaussian states of $(m-1)$ modes to determine $\tilde{P}(n)$.
This calculation can be performed in parallel polynomial time using an exponential number of processors.
The precision error between $\tilde{P}(n)$ and $P(n)$, with $P(n)$ defined as 
\begin{equation}
    P(n) = \Tr[\Pi_n \rho] = \Tr[\ket{n}\bra{n}\otimes \I_{m-1} \rho]
\end{equation}
is given by
\begin{eqnarray}
    |P(n) - \tilde{P}(n)| &=& |\Tr[\Pi_n \rho - \Pi_{\tilde{n}} \rho + \Pi_{\tilde{n}} \rho - \Pi_{\tilde{n}}\tilde{\rho}]| \nonumber \\
    &\leq& |\Tr[\Pi_n \rho - \Pi_{\tilde{n}} \rho]| + |\Tr[\Pi_{\tilde{n}} \rho - \Pi_{\tilde{n}}\tilde{\rho}]| \nonumber \\
    &\leq& D(\Pi_n,\Pi_{\tilde{n}}) + D(\rho,\tilde{\rho}) \nonumber \\
    &\leq& 1/\exp
\end{eqnarray}
$\tilde{P}(n)$ approximates $P(n)$ up to precision $1/\exp$ for $n = \poly$.

Combining the parallel time complexity of obtaining the Gaussian superposition approximation of $\hat{U}_t\ket{0}^{\otimes m}$ via Algorithm \ref{algo:decomp} ($\poly$), the coherent superposition approximation of $\ket{n}$ ($\poly\log$) and the probability estimation using these approximations ($\poly$), we can estimate the required probability $P(n)$ given by Eq.~(\ref{eqn:prob_to_calculate}) up to inverse exponential precision in parallel $\poly$ time. Hence, this implies that the $\textsf{CVBQP}$-complete problem \hyperref[defi:prob_passive]{\textsc{PassiveBosonNumEst}}, which only requires inverse polynomial precision, can be solved in parallel $\poly$ time, thereby establishing Theorem \ref{theo:complexity_prob}, i.e., \hyperref[defi:prob_passive]{\textsc{PassiveBosonNumEst}} $\in$ \textsf{PSPACE}.
By the definition of a complete problem, this also concludes the proof of Theorem \ref{theo:main_result}, i.e., $\textsf{CVBQP} \subseteq \textsf{PSPACE}$.

\subsection{Potential limitations for improved simulation}
\label{sec:PSPACE}

Our main result Theorem \ref{theo:main_result} provides a perspective on the computational power of practical bosonic computations improving the previous bound $\textsf{CVBQP}\subseteq \textsf{EXPSPACE}$ \cite{chabaud2024complexity} to $\textsf{CVBQP}\subseteq \textsf{PSPACE}$. This improvement stems from two key insights. Firstly, by focusing on solving a complete problem rather than imposing stronger simulation constraints, our approach captures the essential complexity of \textsf{CVBQP} while avoiding unnecessary computational overhead. Second, we ensure a reasonably bounded energy throughout the computation (through the use of gentle measurement lemma in Lemma \ref{lem:csr_decomp_cubic}), avoiding the unrestricted energy growth permitted in prior algorithms.

Note that our current proof technique does not allow us to refine the bounds further, e.g., to show that $\textsf{CVBQP} \subseteq \textsf{PP}$. The proof for $\textsf{BQP} \subseteq \textsf{PP}$ relies on an algorithm summing up the number of computational paths leading to the desired output \cite{Adleman1997}. Since each of the computational paths can be simulated in polynomial time and given the finite number of paths for DV systems, this algorithm can solve \textsf{BQP} in \textsf{PP} by picking a random computational path. However, since for CV systems we have an infinite number of computational paths leading to the desired output, such a proof strategy fails. Our algorithm, which instead approximates the output state by a finite superposition of coherent states---essentially corresponding to a finite number of computational paths---based on a finite energy cutoff seems to require exponential space to store the descriptions of the amplitudes in the superposition (to keep the approximation error bounded). Therefore, it is unclear how to generalize it to place \textsf{CVBQP} within a smaller class like \textsf{PP}. If a more efficient method for truncating Airy functions (that does not require evaluating polynomials of exponential degree), crucially used in the calculation of Fock state amplitudes in Lemma \ref{lem:csr_decomp_cubic}, could be obtained, it would enable us to compute Fock state amplitudes in polynomial time.
Then we could use the trick of summing over all ``paths'' of coherent states to potentially compute the acceptance probability as the difference between accepting and rejecting paths of a nondeterministic Turing machine, thereby proving that $\textsf{CVBQP}\subseteq \textsf{PP}$.

The potential exclusion of \textsf{CVBQP} from \textsf{PP} could be due to bosonic systems offering an exponential advantage over discrete-variable systems, as previously suggested, or it may indicate that \textsf{CVBQP} does not accurately characterize realistic bosonic circuits, requiring a focus on circuits with lower average energy. Before drawing conclusions, further efforts should be directed toward simulating \textsf{CVBQP} within \textsf{PP}. The algorithm we propose here still seems to address a more challenging problem than necessary, as it estimates probabilities with exponential precision rather than polynomial precision. This indicates that there is potential for further optimization.


\section{Discussion}\label{sec:conclusion}

We demonstrate that bosonic quantum computations can be simulated in exponential time, significantly improving the previous best-known upper bound of exponential space. This is achieved by formulating a simple problem that encapsulates the complexity of these computations and preventing unrestricted growth of energy throughout the computation.

This paper advances our understanding of the practical computational power of bosonic systems and their viability as quantum computing platforms. We provide strong evidence that \textsf{CVBQP} is a well-motivated complexity class for problems solvable by bosonic quantum computers as it establishes a practical upper bound on their computational power. Further, we show that while bosonic quantum systems may offer a computational advantage over DV quantum computers, this advantage is at most exponential. Although this may still seem like a big difference, the result is a significant improvement over previous results which suggested an implausibly large gap in computational power between the two. With previous results suggesting that bosonic quantum computers might indeed provide exponential advantage compared to their DV counterparts \cite{AnnaMele2024,brenner2024factoring,upreti2025interplay}, our results motivate a rigorous exploration of this question.

Several results of independent interest also arise: our circuit reduction theorem (Lemma \ref{lem:passive_reduction}) allows to remove most squeezing in a circuit without increasing the number of non-Gaussianity-inducing cubic phase gates, improving results from \cite{Sefi2011}. Further, our simulation algorithm (Algorithm \ref{algo:decomp}) highlights coherent state rank as a non-Gaussian resource. The fact that the approximate coherent state rank is finite under the action of cubic phase gates and the coherent state approximation can be described rather simply in the standard model of CVQC underscore its importance in CV circuits, as similar simulation strategies using other non-Gaussian measures such as stellar rank \cite{chabaud2020stellar,chabaud2023resources} or Gaussian rank \cite{hahn2024classicalsimulation} are not trivially applicable.

Our work also opens several research directions.
Firstly, with Gaussian unitary gates and cubic phase gate being universal for CVQC \cite{Lloyd1999,arzani2025} and the availability of exact decomposition methods for many non-Gaussian gates into this gate set \cite{Kalajdzievski2021}, it would be interesting to see what other important CVQC circuits can be decomposed in polynomial number of cubic phase gates and Gaussian unitaries, to give a larger class of bosonic computations that can be simulated in exponential time. 

Moreover, a natural open question is whether $\textsf{CVBQP} \subseteq \textsf{PP}$, which would bring \textsf{CVBQP} even closer to $\textsf{BQP}$. In contrast, another compelling direction would be to investigate whether such a result could cause a collapse of the polynomial hierarchy, suggesting bosonic platforms as a theoretically preferred alternative to DV quantum computers.
Furthermore, one could consider defining $\textsf{CVBQP}$ with a more stringent energy constraint, as energy must be supplied to a circuit in any practical scenario and is inherently finite. In this regard, it would be interesting to explore the energy constraint that makes $\textsf{CVBQP} = \textsf{BQP}$ and whether experimental scenarios could provide more energy than this constraint. In other terms, what is the precise trade-off between energy and time/space complexity when computing with quantum systems? We leave these questions for future research.



\section{Methods}\label{sec:methods}
This section gives a few definitions relevant to the paper and additional technical results useful in the proof of Theorem \ref{theo:main_result}.
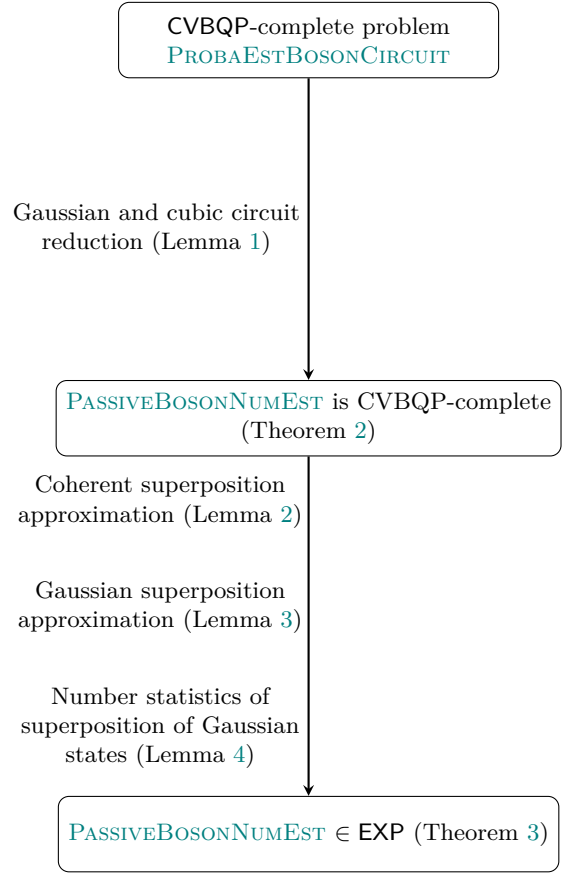
\begin{figure}[h!]
\centering
\begin{tikzpicture}[node distance=4 cm]

\node (proba) [process] {\textsf{CVBQP}-complete problem \\ \hyperref[defi:org_CVBQP]{\textsc{ProbaEstBosonCircuit}}};
\node (theorem2) [process, below=of proba] {\hyperref[defi:prob_passive]{\textsc{PassiveBosonNumEst}} is CVBQP-complete \\ (Theorem \ref{theo:CVBQP_complete_desc})};
\node (theorem3) [process, below=4.5cmof theorem2] {\hyperref[defi:prob_passive]{\textsc{PassiveBosonNumEst}} $\in\mathsf{PSPACE}$ (Theorem \ref{theo:complexity_prob})};

\draw [arrow] (proba) -- (theorem2);
\draw [arrow] (theorem2) -- (theorem3);

\node (gaussian) at ([xshift=-2 cm, yshift=-2cm] proba.south) [align=center, font = \fontsize{9}{11}\selectfont] {Gaussian and cubic circuit \\ reduction (Lemma \ref{lem:passive_reduction})};
\node (lemma2) at  ([xshift=-1.95 cm, yshift=-0.6cm] theorem2.south) [align = center, font = \fontsize{9}{11}\selectfont] {Coherent superposition \\ approximation (Lemma \ref{lem:csr_decomp_cubic})};
\node (lemma3) at  ([xshift=-1.95 cm, yshift=-2cm] theorem2.south) [align = center, font = \fontsize{9}{11}\selectfont] {Gaussian superposition \\ approximation (Lemma \ref{lem:csr})};
\node (lemma4) at  ([xshift=-1.95 cm, yshift=-3.6cm] theorem2.south) [align = center, font = \fontsize{9}{11}\selectfont] {Number statistics of \\superposition of Gaussian \\ states (Lemma \ref{lem:prob})};

\end{tikzpicture}
\caption{Proof sketch of Theorem \ref{theo:main_result}. Starting with the definition of \textsf{CVBQP}-complete problem \hyperref[defi:org_CVBQP]{\textsc{ProbaEstBosonCircuit}} \cite{chabaud2024complexity}, Gaussian circuit reduction (Lemma \ref{lem:passive_reduction}) allows us to define a simple problem \hyperref[defi:prob_passive]{\textsc{PassiveBosonNumEst}} that captures the complexity of \textsc{CVBQP} computations (Theorem \ref{theo:CVBQP_complete_desc}). Using Coherent superposition approximation(Lemma \ref{lem:csr_decomp_cubic}), Gaussian superposition approximation (Lemma \ref{lem:csr}) and Number statistics of superposition of Gaussian states (Lemma \ref{lem:prob}), we prove \hyperref[defi:prob_passive]{\textsc{PassiveBosonNumEst}} $\in \textsf{PSPACE}$ (Theorem \ref{theo:complexity_prob}), and consequently $\textsf{CVBQP} \subseteq \textsf{PSPACE}$ (Theorem \ref{theo:main_result}).}
\label{fig:proof_sketch}
\end{figure}

\subsection{Closeness of pure quantum states}

We use three different notions of closeness of two pure quantum states $\ket{\psi}$ and $\ket{\phi}$ in the paper: trace distance ($D$), fidelity ($F$) and $2$-norm distance.

The trace distance between two pure states is given by
\begin{equation}
    D(\psi,\phi) = \frac12 \|\ket{\psi}\bra{\psi} - \ket{\phi}\bra{\phi}\|_1,
\end{equation}
where $\|.\|_1$ refers to the operator $1$-norm. It is related to fidelity $F(\psi,\phi) = |\braket{\phi|\psi}|^2$ for pure states as
\begin{equation}
    D(\psi,\phi) = \sqrt{1 - F(\psi,\phi)}.
\end{equation}
Finally, the $2$-norm distance is given as
\begin{equation}
\|\ket\psi - \ket{\phi}\| = \sqrt{(\bra\psi - \bra{\phi})(\ket\psi - \ket{\phi})},
\end{equation}
where $\|\cdot\|$ refers to the $2$-norm. It is related to the trace distance as
\begin{equation}
    D(\psi,\phi) \leq \|\ket\psi - \ket{\phi}\|,
\end{equation}
since
\begin{equation}
    \|\ket\psi - \ket{\phi}\|^2 = 2(1 - \Re(\braket{\phi|\psi})) \geq 2(1 - \sqrt{F(\psi,\phi)}),
\end{equation}
and
\begin{equation}
    2(1 - \sqrt{F(\psi,\phi)}) \geq 1 - F(\psi,\phi) = D(\psi,\phi)^2.
\end{equation}

\subsection{Complexity classes} 

We give the definition of the complexity classes relevant for the paper \cite{papadimitriou1994computational,Aaronson2013book,chabaud2024complexity}. These classes are defined for decision problems, that is, problems whose answer is yes or no. For example, for a given input bit string $x$ of size $m$, ``Is $x$ prime?'' defines a decision problem. For bosonic quantum computers, the input size $m$ equivalently refers to the number of modes $m$. The scaling of resources is determined with respect to this input size $m$.
    
\textsf{EXPSPACE} consists of decision problems solvable by a deterministic Turing machine using exponential space, with no time restrictions.

\textsf{EXP} (or \textsf{EXPTIME}) is the class of decision problems that can be solved by a deterministic Turing machine in exponential time, i.e., time bounded by $\exp^{p(m)}$, where $p(m)$ is a polynomial function of the input size $m$.

\textsf{PSPACE} consists of decision problems solvable by a deterministic Turing machine using polynomial space, with no time restrictions.

The complexity class \textsf{PP} (Probabilistic Polynomial time) includes decision problems that can be solved by a probabilistic algorithm running in polynomial time on a Turing machine (a classical computer), where the probability of producing the correct answer exceeds 50\%. It characterizes problems where a computational process with random guessing has a slight advantage in determining the correct solution.

\textsf{BQP} (Bounded-Error Quantum Polynomial time) is the class of decision problems solvable by a DV quantum computer in polynomial time with a bounded probability of error. Formally, a problem belongs to \textsf{BQP} if there exists a quantum algorithm that outputs the correct answer with probability at least $2/3$ and runs in time polynomial in the input size $m$.

Finally, \textsf{CVBQP} (Bounded error Continuous-Variable Quantum Polynomial time computations) is the class of decision problems solvable by bosonic (CV) quantum computers with number of Gaussian unitary and non-Gaussian cubic phase gates polynomial in input size $m$ (problem solvable by bosonic quantum computers in polynomial time) with a bounded probability of error, the probability of correct answer being at least $2/3$.

With these definitions, while a relation like $\textsf{EXP}\subseteq \textsf{EXPSPACE}$ suggests that there are problems in $\textsf{EXPSPACE}$ that are more difficult to solve than all problems in $\textsf{EXP}$, relations like $\textsf{BQP} \subseteq \textsf{PP}$ imply that DV quantum computers cannot solve problems more complex than those belonging to the class $\textsf{PP}$. Finally, $\textsf{BQP} \subseteq \textsf{CVBQP}$ \cite{chabaud2024complexity} implies that DV quantum computers cannot solve harder problem than bosonic quantum computers\footnote{This is formally known for a bosonic gate set without the cubic phase gate \cite{chabaud2024complexity} but is expected to hold for that set too \cite{Gottesman2001,Sefi2011,arzani2025}.}. Whether this is an equality remains an open question, even when assuming promised energy bounds for the \textsf{CVBQP} computations.

\subsection{Approximation of $e^{i\gamma\hat{q}^3}\ket{\alpha}$}

A key component in the approximation of $\ket{\psi_\mathrm{out}} = \hat{U}_t \ket{0}^{\otimes m}$ as a superposition of Gaussian states is the approximation of $e^{i\gamma\hat{q}^3}\ket{\alpha}$ as a superposition of coherent states up to some trace distance error $\epsilon$ (see Algorithm \ref{algo:decomp}). This is always possible with the following result:
\begin{restatable}[$\epsilon$-approximate coherent state rank of $e^{i\gamma\hat{q}^3}\ket{\alpha}$]{lem}{lemcsrdecompcubicweak}\label{lem:csr_decomp_cubic_weak} For $\alpha \in \C, \gamma \in \R$, given a state of finite energy $\ket{\psi_{\alpha,\gamma}} = e^{i\gamma\hat{q}^3}\ket{\alpha}$, the $\epsilon$-approximate coherent state rank of $\ket{\psi_{\alpha,\gamma}}$ is upper bounded by $4E(\alpha,\gamma)/\epsilon^2 + 1$, where $E(\alpha,\gamma)$ is the energy of $e^{i\gamma\hat{q}^3}\ket{\alpha}$, which is a polynomial in $\alpha$, $\alpha^*$ and $\gamma$.
\end{restatable}
\noindent The proof of Lemma \ref{lem:csr_decomp_cubic_weak} is given in section \ref{appendixsec:csr_decomp_cubic_weak} of the appendix and involves using the gentle measurement lemma \cite{Wilde2013} to truncate $\ket{\psi_{\alpha,\gamma}}$ to a superposition of Fock states of finite size, depending on the energy of $\ket{\psi_{\alpha,\gamma}}$ and the required precision $\epsilon$. Then applying Theorem 1 from \cite{Marshall2023} allows us to find the coherent state decomposition of the truncated state.

However, Lemma \ref{lem:csr_decomp_cubic_weak} is not enough for our purpose, for two reasons. Firstly, even though it proves the existence of a coherent state superposition approximation for $\ket{\psi_{\alpha,\gamma}}$, it is not obvious how to compute it a priori, and with which time complexity, given the required precision. Secondly, the $\epsilon$-approximate coherent state rank quantifies closeness in terms of the trace distance, i.e., closeness up to a global phase. While this global phase is irrelevant for a single coherent state, in our circuit, we are typically dealing with superpositions of coherent states as an input to a cubic phase gate. When approximating the action of the cubic phase gate on each coherent state using Lemma \ref{lem:csr_decomp_cubic_weak}, the resulting phases introduce relative phase shifts between the superposition terms, potentially leading to significant trace distance errors. To address this, we require a decomposition of \(\ket{\psi_{\alpha,\gamma}}\) that can be efficiently computed within the desired precision and is based on a stronger notion of closeness than trace distance. Lemma \ref{lem:csr_decomp_cubic} provides precisely this refinement.

\subsection{Parallel computation}

Implementing Algorithm~\ref{algo:decomp} on a Turing machine directly gives exponential runtime and exponential space.
To achieve polynomial space complexity, we make use of \emph{parallel computation}.
It turns out that $\PSPACE$ can equivalently be defined as the class of problems that can be solved by a polynomial-time computation with an exponential number of processors in the PRAM model \cite{FW78}, or polynomial-depth uniform circuits \cite{Bor77} with an exponential number of gates.
In that case, the output of each gate can be computed using a recursive algorithm requiring only polynomial space.
Since \textsf{CVBQP} circuits have polynomial depth, the main difficulty is to parallelize the application of Lemma~\ref{lem:csr_decomp_cubic}.
Firstly, observe that the coherent state amplitudes can be computed independently and thus in parallel.
Then, to compute the coefficients of the truncated state up to doubly exponential precision, we use the fact that arithmetic including solving polynomial systems has polylogarithmic parallel complexity in the bit length of the inputs \cite{BCH86,Nef94} (see details in Appendix~\ref{sec:parallel}).

\hspace{5mm}


\section{Acknowledgements}

U.C.\ acknowledges inspiring discussions with J.\ Marshall, A.\ Motamedi, S.\ Mehraban, and S.\ Gharibian. We acknowledge funding from the European Union’s Horizon Europe Framework Programme (EIC Pathfinder Challenge project Veriqub) under Grant Agreement No.\ 101114899.

\onecolumn

\bibliographystyle{alpha}
\bibliography{sample}
\appendix
\section{Proof of Theorem \ref{theo:CVBQP_complete_desc}}\label{appendixsec:CVBQP_complete_desc}
\noindent First, we formally define \textsf{CVBQP}:
\begin{defi}
    [\textsf{CVBQP}] \cite[Definition 4.4]{chabaud2024complexity} \textsf{CVBQP} is the class of languages $L \subseteq \{0,1\}^*$ for which there exists a bosonic quantum circuit, described by the unitary $\hat{U}_t$, with $\poly$ number of Gaussian unitary and cubic phase gates, with the elements of the Gaussian unitary gates and cubicity of the cubic phase gates specified upto $m$ bits of precision, and single-mode boson number measurement, such that
    \begin{itemize}
        \item for all $x \in L$, $\mathcal{C}_{|x|}(x) \in [b,b+E]$ (accepts) with probability greater than $\frac{2}{3}$,
        \item for all $x \notin L$, $\mathcal{C}_{|x|}(x) \in [a-E,a]$ (rejects) with probability greater than $\frac{2}{3}$,
        \end{itemize}
given constants $a,b,E = \poly$ such that $b-a \geq \Omega\left(\frac{1}{\poly}\right)$ and $a - E \geq 0$, and $\mathcal{C}_{|x|}(x)$ denotes the output of the single-mode number measurement on $\hat{U}_t \ket{x}$, $\ket{x}$ representing the coherent state.
\end{defi}
\noindent A $\textsf{CVBQP}$-complete problem is \cite[Definition 4.5]{chabaud2024complexity}:
  \begin{defi}\label{defi:org_CVBQP}
      (\textsc{ProbaEstBosonCircuit}) Given an input vacuum state $\ket{0}^{\otimes m}$ and description of a circuit $\hat{U}_t$ with $\poly$ cubic phase gates and Gaussian unitary gates, with the elements of the Gaussian unitary gates and cubicity of the cubic phase gates specified upto $m$ bits of precision, and single mode number measurement, \textsf{CVBQP}-complete problem is the problem of deciding whether the probability of boson number measurement on the first mode of the circuit, $\mathrm{Pr}\left[\mathcal{C}(0) \in [b,b+E]\right] > 2/3$ or $\mathrm{Pr}\left[\mathcal{C}(0) \in [a-E,a]\right] > 2/3$, provided that one of the two cases hold, and such that $a,b,E = \poly$, $a - E \geq 0$, and $b - a = \Omega\left(\frac{1}{\poly}\right)$, and $\mathcal{C}(0)$ denotes the output of the single-mode number measurement on $\hat{U}_t \ket{0}^{\otimes  m}$.
  \end{defi}
\noindent With these definitions, we restate Theorem \ref{theo:CVBQP_complete_desc}:
\thmCVBQPcomplete*
\begin{proof}
Since \hyperref[defi:prob_passive]{\textsc{PassiveBosonNumEst}} is an instance of the \textsf{CVBQP}-complete problem \hyperref[defi:org_CVBQP]{\textsc{ProbaEstBosonCircuit}}, proving that it is a \textsf{CVBQP}-hard problem is enough to prove that is \textsf{CVBQP}-complete.

If we can estimate $\mathrm{Pr}\left[\mathcal{C}(0) \in [b,b+E]\right]$ or $\mathrm{Pr}\left[\mathcal{C}(0) \in [a-E,a]\right]$ up to additive precision $1/\poly$, then the mutual exclusivity of the two regions guarantee that we can answer the \textsf{CVBQP}-complete problem. Giving the circuit description
\begin{equation}\label{eqn:U_t_org}
    \hat{U}_t = \hat{G}_t e^{i\gamma_t\hat{q}_1^3} \hat{G}_{t-1} \dots e^{i\gamma_1\hat{q}_1^3} \hat{G}_{0}
\end{equation}
acting on the input vacuum state, where $\hat{G}_0,\dots,\hat{G}_t$ are Gaussian unitary gates, with $t = \mathcal{O}(\poly(m))$, if we can calculate
\begin{equation}
    P(n) = \Tr[\ket{n}\bra{n}\otimes \I_{m-1} \hat{U}_t \ket{0} \bra{0}^{\otimes m} \hat{U}_t^\dagger]
\end{equation}
up to additive precision $1/\poly$ for $n = \poly$, we can calculate $\mathrm{Pr}\left[\mathcal{C}(0) \in [b,b+E]\right]$ for $b,E = \poly$ up to total variational distance $1/\poly$ and solve the $\textsf{CVBQP}$-complete problem. Note that cubic phase gates acting on other modes can be swapped into the first mode and the resulting SWAP gate can be absorbed into the Gaussian unitary gates.

Now we rewrite $\hat{U}_t$ in Eq.~(\ref{eqn:U_t_org}) in a form which aids the further analysis of the problem. First, by inserting identity $\I = \hat{G}^\dagger \hat{G}$ with the appropriate Gaussian $\hat{G}$ and denoting
\begin{eqnarray}\label{eqn:Gaussian_brevity}
    \hat{\mathcal{G}}_i &:=& \hat{G}_i\hat{G}_{i-1}\dots\hat{G}_0, \hspace{0.5mm} \forall i \in \{0,\dots,m\}
\end{eqnarray}
for brevity, we obtain 
\begin{equation}
    \hat{U}_t = \hat{\mathcal{G}}_t \hat{\mathcal{G}}_{t-1}^\dagger e^{i\gamma_t\hat{q}_1^3}  \hat{\mathcal{G}}_{t-1} \dots \hat{\mathcal{G}}_{1}^\dagger e^{i\gamma_1\hat{q}_1^3}  \hat{\mathcal{G}}_{1} \hat{\mathcal{G}}_{0}^\dagger e^{i\gamma_1\hat{q}_1^3}  \hat{\mathcal{G}}_{0} .
\end{equation}
Then applying Lemma \ref{lem:passive_reduction} (see Appendix~\ref{appendixsec:proof_lem_passive_reduction} for a proof) and denoting
\begin{eqnarray}\label{eqn:passive_notations}
    \hat{B}_0 &:=& \hat{V}_0, \hspace{5mm} \mathcal{G}_t \hat{B}^\dagger_{t-1}  := \hat{G},\nonumber \\
    \hat{B}_{i} \hat{B}^\dagger_{i-1}  &:=& \hat{V}_i\quad\forall i \in \{1,2,\dots,t-1 \},
\end{eqnarray}
we get the unitary
\begin{equation}
    \hat{U}_t = \hat{G} e^{i\gamma_t'\hat{q}_1^3} \hat{V}_{t-1} \dots \hat{V}_1 e^{i\gamma_1'\hat{q}_1^3}\hat{V}_0,
\end{equation}
where $ \hat{V}_{t-1} \dots \hat{V}_0$ are displaced passive linear unitaries and $\hat{G}$ is a Gaussian unitary. Therefore, \hyperref[defi:prob_passive]{\textsc{PassiveBosonNumEst}} is $\mathsf{CVBQP}$-hard and since it is an instance of the \textsf{CVBQP}-complete problem \hyperref[defi:org_CVBQP]{\textsc{ProbaEstBosonCircuit}}, \hyperref[defi:prob_passive]{\textsc{PassiveBosonNumEst}} is $\mathsf{CVBQP}$-complete.
\end{proof}

\section{Proof of Lemma \ref{lem:passive_reduction}}\label{appendixsec:proof_lem_passive_reduction}
\noindent We restate Lemma \ref{lem:passive_reduction}:
\lemPassiveReduction*
\begin{proof}
    The proof works by inserting passive linear unitaries to compress the effect of $\hat{G}$ to the first quadrature.
\begin{equation}
    \hat{G}^\dagger e^{i\gamma \hat{q}^3} \hat{G} = e^{i\gamma (\hat{G}^\dagger \hat{q_1} \hat{G})^3}.
\end{equation}
Now,
\begin{equation}
    \hat{G}^\dagger \hat{q}_1 \hat{G} = (S)_{q_1,q_1} \hat{q}_1 + \dots + (S)_{q_1,q_m} \hat{q}_m + (S)_{q_1,p_1} \hat{p}_1 + \dots + (S)_{q_1,p_m} \hat{p}_m + d_1,
 \end{equation}
 so
 \begin{eqnarray}
     \hat{G}^\dagger \hat{q}_1 \hat{G} &=& \frac{\sqrt{(S)_{q_1,q_1}^2 + (S)_{q_1,p_1}^2}}{\sqrt{(S)_{q_1,q_1}^2 + (S)_{q_1,p_1}^2}} ((S)_{q_1,q_1} \hat{q}_1 + (S)_{q_1,p_1} \hat{p}_1) + \dots \nonumber\\ & &+ \frac{\sqrt{(S)_{q_1,q_m}^2 + (S)_{q_1,p_m}^2}}{\sqrt{(S)_{q_1,q_m}^2 + (S)_{q_1,p_m}^2}} ((S)_{q_1,q_m} \hat{q}_t + (S)_{q_1,p_m} \hat{p}_m) + d_1.
 \end{eqnarray}
Inserting identity $\I = \hat{D}^\dagger(d_1)\hat{D}(d_1)$ on either side of the Gaussians,
 \begin{eqnarray}
     \hat{D}(d_1)\hat{G}^\dagger \hat{q}_1 \hat{G} \hat{D}^\dagger(d_1) &=& \frac{\sqrt{(S)_{q_1,q_1}^2 + (S)_{q_1,p_1}^2}}{\sqrt{(S)_{q_1,q_1}^2 + (S)_{q_1,p_1}^2}} ((S)_{q_1,q_1} \hat{q}_1 + (S)_{q_1,p_1} \hat{p}_1) + \dots \nonumber\\ & &+ \frac{\sqrt{(S)_{q_1,q_m}^2 + (S)_{q_1,p_m}^2}}{\sqrt{(S)_{q_1,q_m}^2 + (S)_{q_1,p_m}^2}} ((S)_{q_1,q_m} \hat{q}_m + (S)_{q_1,p_m} \hat{p}_m).
 \end{eqnarray}
 If we modify the circuit by adding rotation gates $\hat{R}(-\theta_1)\otimes \dots \otimes \hat{R}(-\theta_m) = \hat{R}^\dagger(\theta_1)\otimes \dots \otimes \hat{R}^\dagger(\theta_m)$, where $\theta_i$ such that $\cos(\theta_i) = \frac{(S)_{q_1,q_i}}{\sqrt{(S)_{q_1,q_i}^2 + (S)_{q_1,p_i}^2}}$ and $\sin(\theta_i) = \frac{(S)_{q_1,p_i}}{\sqrt{(S)_{q_1,q_i}^2 + (S)_{q_1,p_i}^2}}$, where $i \in \{1,\dots,m\}$. Then
 \begin{eqnarray}
    && \hspace{-30mm} \hat{R}(\theta_m) \otimes \dots \otimes \hat{R}(\theta_1) \hat{D}(d_1)\hat{G}^\dagger \hat{q}_1 \hat{G} \hat{D}^\dagger(d_1) \hat{R}^\dagger(\theta_1)\otimes \dots \otimes \hat{R}^\dagger(\theta_m) \nonumber \\ && \hspace{-20mm} = \sqrt{(S)_{q_1,q_1}^2 + (S)_{q_1,p_1}^2} \hat{q}_1 + \dots +  \sqrt{(S)_{q_1,q_t}^2 + (S)_{q_1,p_t}^2} \hat{q}_t\nonumber \\ && \hspace{-20mm}= \frac{\sqrt{(S)_{q_1,q_1}^2 + (S)_{q_1,p_1}^2 + \dots + (S)_{q_1,q_t}^2 + (S)_{q_1,p_t}^2}}{\sqrt{(S)_{q_1,q_1}^2 + (S)_{q_1,p_1}^2 + \dots + (S)_{q_1,q_t}^2 + (S)_{q_1,p_t}^2}} \nonumber \\ && \times \left(\sqrt{(S)_{q_1,q_1}^2 + (S)_{q_1,p_1}^2} \hat{q}_1 + \dots +  \sqrt{(S)_{q_1,q_m}^2 + (S)_{q_1,p_m}^2} \hat{q}_t \right).
 \end{eqnarray}
Modifying again the circuit by adding a passive unitary $\hat{O}$ such that transforms the quadratures such that
 \begin{equation}
     \hat{q}_1 \rightarrow \frac{\left(\sqrt{(S)_{q_1,q_1}^2 + (S)_{q_1,p_1}^2} \hat{q}_1 + \dots +  \sqrt{(S)_{q_1,q_m}^2 + (S)_{q_1,p_m}^2} \hat{q}_m \right)}{\sqrt{(S)_{q_1,q_1}^2 + (S)_{q_1,p_1}^2 + \dots + (S)_{q_1,q_m}^2 + (S)_{q_1,p_m}^2}},
 \end{equation}
 we have
 \begin{eqnarray}
     && \hspace{-25mm} \hat{O} \hat{R}(\theta_m) \otimes \dots \otimes \hat{R}(\theta_1) \hat{D}(d_1)\hat{G}^\dagger \hat{q}_1 \hat{G} \hat{D}^\dagger(d_1) \hat{R}^\dagger(\theta_1)\otimes \dots \otimes \hat{R}^\dagger(\theta_m) \hat{O}^\dagger \nonumber\\ && = \sqrt{(S)_{q_1,q_1}^2 + (S)_{q_1,p_1}^2 + \dots + (S)_{q_1,q_m}^2 + (S)_{q_1,p_m}^2} \hat{q}_1 := \|(S)_1\|_2 \hat{q}_1.
 \end{eqnarray}
Therefore,
 \begin{equation}
     \hat{G}^\dagger e^{i\gamma \hat{q}_1^3} \hat{G} = \hat{D}^\dagger(d_1)\hat{R}^\dagger(\theta_m)\otimes \dots \otimes \hat{R}^\dagger(\theta_1) \hat{O}^\dagger e^{i\gamma \|(S)_1\|_2^3 \hat{q}_1^3} \hat{O} \hat{R}(\theta_1)\otimes \dots \otimes \hat{R}(\theta_m) \hat{D}(d_1),
 \end{equation}
where $\hat{D}(d_1)$ is a single-mode displacement operator, which can equivalently be written as $\hat{D}(d_1) \otimes \I_{m-1}$. Setting
\begin{eqnarray}
    \hat{P} &:=& \hat{O} \hat{R}(\theta_1)\otimes \dots \otimes \hat{R}(\theta_m), \\
    \gamma' &:=& \gamma \|(S)_1\|_2^3,
\end{eqnarray}
we have
\begin{equation}
    \hat{G}^\dagger e^{i\gamma \hat{q}_1^3} \hat{G} = \hat{D}^\dagger(d_1) \otimes \I_{m-1} \hat{P}^\dagger e^{i\gamma' \hat{q}_1^3} \hat{P} \hat{D}(d_1) \otimes \I_{m-1}.
\end{equation}
Setting
\begin{equation}
    \hat{B} := \hat{P} \hat{D}(d_1) \otimes \I_{m-1},
\end{equation}
we conclude
\begin{equation}
    \hat{G}^\dagger e^{i\gamma \hat{q}_1^3} \hat{G} = \hat{B}^\dagger e^{i\gamma' \hat{q}_1^3} \hat{B}.
\end{equation}
\end{proof}
\section{Proof of Lemma \ref{lem:csr_decomp_cubic}}\label{appendixsec:proof_csr_decomp_cubic}
\noindent We restate Lemma \ref{lem:csr_decomp_cubic}:
\lemcsrdecompcubic*
\begin{proof}
We find a coherent state approximation of $\ket{\psi_{\alpha,\gamma}} = e^{i\gamma\hat{q}_1^3} \ket{\alpha}$ in three steps:
\begin{itemize}
    \item Using the gentle measurement lemma, we truncate $\ket{\psi_{\alpha,\gamma}}$ to a state $\ket{\psi_N'}$ with a finite Fock support size $N+1$, such that 
    \begin{equation}
    \|\ket{\psi_{\alpha,\gamma}}\bra{\psi_{\alpha,\gamma}} - \ket{\psi_N'} \bra{\psi_N'}\|_1 \leq 1/\exp.
    \end{equation}
    \item We approximate the Fock state amplitudes of $\ket{\psi_N'}$ to precision $1/\eexp$ to obtain $\ket{\psi_N}$ and ensure that the 2-norm distance satisfies:
    \begin{equation}
        \|\ket{\psi_N'} - \ket{\psi_N}\| \leq 1/\eexp.
    \end{equation}
    \item We find a coherent state approximation of $\ket{\psi_N}$ using \cite[Theorem 1]{Marshall2023}: given $\ket{\psi_N} = \sum_{i=1}^N a_n \ket{n}$, there is a coherent state decomposition
    \begin{eqnarray}
            \ket{\tilde{\psi}_{\alpha,\gamma}} &=& \frac{1}{\sqrt{\mathcal{N}}} \sum_{k=0}^{N} c_k \ket{\tilde{\epsilon} e^{2\pi ik/(N+1)}}, \\
            c_k &=& \frac{e^{\tilde{\epsilon}^2/2}}{N+1} \sum_{n=0}^N \sqrt{n!} \frac{a_n}{\tilde{\epsilon}^n} e^{-2\pi ink/(N+1)},
        \end{eqnarray}
    where the normalization factor satisfies $\mathcal{N} = 1 + \mathcal{O}(\tilde{\epsilon}^{2(N+1)}/(N+1)!)$, and $\ket{\tilde{\psi}_{\alpha,\gamma}}$ is $\delta$-close in trace distance to $\ket{\psi_N}$, with $\delta = \mathcal{O}(\tilde{\epsilon}^{(N+1)}/\sqrt{(N+1)}!)$. Choosing $\delta = 1/\exp$ gives $\ket{\tilde{\psi}_{\alpha,\gamma}}$ that is $1/\exp$ close in trace distance to $\ket{\psi_{\alpha,\gamma}}$, and by adding a global phase to $\ket{\tilde{\psi}_{\alpha,\gamma}}$, we can further ensure $\ket{\tilde{\psi}_{\alpha,\gamma}}$ is $1/\exp$ close to $\ket{\psi_{\alpha,\gamma}}$ in 2-norm distance as well.
\end{itemize}

\medskip

We first use the gentle measurement lemma to truncate $\ket{\psi_{\alpha,\gamma}}$ to a state with finite support size in Fock basis:
\begin{lem}[Gentle measurement lemma \cite{Wilde2013}]\label{lem:gml}
     Given a state 
    \begin{equation}
        \ket{\psi_{\alpha,\gamma}} = e^{i\gamma\hat{q}_1^3}\ket{\alpha},
    \end{equation}
    such that 
    \begin{equation}
        \bra{\psi_{\alpha,\gamma}}\hat{n}\ket{\psi_{\alpha,\gamma}} \leq E,
    \end{equation}
    the normalized projection of $\ket{\psi_{\alpha,\gamma}}$ on the first $N$ states, $\ket{\psi_N'}$ is such that 
    \begin{equation}
        \frac12\|\ket{\psi_{\alpha,\gamma}}\bra{\psi_{\alpha,\gamma}} - \ket{\psi_N'} \bra{\psi_N'}\|_1 \leq \sqrt{\frac{E}{N}},
    \end{equation}
    where $\|.\|_1$ refers to the 1-norm.
\end{lem}
\noindent Therefore, choosing $N = 4E/\epsilon^2$, we find 
\begin{equation}\label{eqn:truncated}
    \ket{\psi_N'} = \frac{1}{\sqrt{\mathcal{N}_N'}}\sum_{i=0}^N a_n \ket{n},
\end{equation}
with $a_n = \bra{n}e^{i\gamma \hat{q}^3}\ket{\alpha}$ and $\mathcal{N}_N' = \sum_{i=0}^N |a_n|^2$ such that
\begin{equation}
   \frac12 \|\ket{\psi_{\alpha,\gamma}}\bra{\psi_{\alpha,\gamma}} - \ket{\psi_N'} \bra{\psi_N'}\| \leq \frac{\epsilon}{2}.
\end{equation}
For $\ket{\psi_{\alpha,\gamma}} = e^{i\gamma\hat{q}^3}\ket{\alpha}$, $E$ is given by
\begin{eqnarray}
    E = \bra{\alpha}e^{-i\gamma\hat{q}^3} \hat{n} e^{i\gamma\hat{q}^3} \ket{\alpha} &=& \bra{\alpha}e^{-i\gamma\hat{q}_1^3} \left(\frac12 (\hat{q}^2 + \hat{p}^2) - \frac12 \right) e^{i\gamma\hat{q}_1^3}  \ket{\alpha} \nonumber \\
    &=& \frac12 \bra{\alpha}\hat{q}^2 + (\hat{p} + 3\gamma\hat{q}^2)^2 \ket{\alpha} - \frac12 \nonumber \\
    &=&\! \frac12 \bra{\alpha} \hat{q}^2\! + \!\hat{p}^2 \ket{\alpha}\! -\! \frac12\! +\! \frac92 \gamma^2 \bra{\alpha} \hat{q}^4 \ket{\alpha} + 3\gamma \bra{\alpha} \hat{p}\hat{q}^2 \ket{\alpha} + 3\gamma \bra{\alpha}\hat{q}^2 \hat{p} \ket{\alpha}.
\end{eqnarray}
Noting that with our conventions, $\hat{D}(\alpha)^\dagger \hat{q} \hat{D}(\alpha) = \hat{q} + \sqrt{2}\Re{\alpha}$, $\hat{D}(\alpha)^\dagger \hat{p} \hat{D}(\alpha) = \hat{p} + \sqrt{2}\Im{\alpha}$, we have
\begin{eqnarray}
    \bra{\alpha}\hat{q}^4\ket{\alpha} &=& 3 + 12\Re{\alpha}^2 + 4\Re{\alpha}^4, \nonumber \\
    \bra{\alpha} \hat{p}\hat{q}^2 \ket{\alpha} &=& 2\sqrt{2}\Re{\alpha} \bra{0}\hat{q}\hat{p}\ket{0} + \sqrt{2}\Im{\alpha} + 2\sqrt{2}\Re{\alpha}^2\Im{\alpha}, \nonumber \\
    \bra{\alpha}\hat{q}^2\hat{p} \ket{\alpha} &=& 2\sqrt{2}\Re{\alpha} \bra{0}\hat{p}\hat{q}\ket{0} + \sqrt{2}\Im{\alpha} + 2\sqrt{2}\Re{\alpha}^2\Im{\alpha}    ,
\end{eqnarray}
so
\begin{equation}\label{eqn:energy}
    E= |\alpha|^2 + \frac{9}{2}\gamma^2(3 + 12\Re{\alpha}^2 + 4\Re{\alpha}^4) + 6\gamma(\sqrt{2}\Im{\alpha} + 2\sqrt{2}\Re{\alpha}^2\Im{\alpha}),
\end{equation}
which is polynomial in $\alpha$, $\alpha^*$ and $\gamma$. Therefore, $\ket{\psi_N'}$ given by Eq.~(\ref{eqn:truncated}) with $N = 4E/\epsilon^2 + 1= 4\poly(\alpha,\gamma)/\epsilon^2 + 1 = \exp $, for $|\alpha|,|\gamma| = \exp$ and $\epsilon = 1/\exp$, approximates $\ket{\psi_{\alpha,\gamma}}$ up to trace distance $1/\exp$.

\medskip

Now we give the method to compute the Fock state coefficients $a_n$ up to precision $1/\eexp$ and the time required to do so.
\begin{eqnarray}
    a_n (\alpha,\gamma) &=& \bra{n}e^{i\gamma\hat{q}^3}\ket{\alpha} \nonumber \\
    &=& \int^{\infty}_{-\infty} dq \braket{n|q}\bra{q}e^{i\gamma\hat{q}^3}\ket{\alpha} \nonumber \\
    &=& \frac{1}{\sqrt{\pi 2^n n!}} \int^{\infty}_{-\infty} dq e^{i\gamma q^3} e^{-q^2/2} H_n(q) e^{-\alpha_I^2} e^{-\frac12(q-\sqrt{2}\alpha)^2} \nonumber \\
    &=& \frac{e^{-\alpha_I^2-\alpha^2}}{\sqrt{\pi 2^n n!}} \int^{\infty}_{-\infty} dq e^{i\gamma q^3} e^{-q^2 + \sqrt{2}q\alpha}H_n(q),
\end{eqnarray}
where $\alpha_I := \Im{\alpha}$ is the imaginary part of $\alpha$.

Following \cite[Appendix D]{Miatto2020fastoptimization}, we perform the substitution of variables $q = yk - iy^3$, such that $y^3 = \frac{1}{3\gamma}$, and obtain
\begin{equation}
    a_n (\alpha,\gamma) = \frac{ye^{\frac23 y^6} e^{-\alpha_I^2-\alpha^2 - i\sqrt{2}\alpha y^3}} {\sqrt{\pi 2^{n} n!}} \int_{-\infty}^{\infty} dk \exp\left[i(y^4 - i\sqrt{2}\alpha y)k + \frac{ik^3}{3} \right] H_n(yk - iy^3).
\end{equation}
Using $H_n(x+y) = \sum_{j=0}^n \binom{n}{j} H_n(x)(2y)^{n-j}$,
\begin{eqnarray}
    a_n (\alpha,\gamma) &=& \frac{ye^{\frac23 y^6} e^{-\alpha_I^2-\alpha^2 - i\sqrt{2}\alpha y^3}} {\sqrt{\pi 2^{n} n!}} \sum_{j=0}^n \binom{n}{j} H_j(-iy^3) (2y)^{n-j} \int_{-\infty}^{\infty} dk \exp\left[i(y^4 - i\sqrt{2}\alpha y)k + \frac{ik^3}{3} \right] (k)^{n-j} \nonumber \\
     &=& \frac{ye^{\frac23 y^6} e^{-\alpha_I^2-\alpha^2 - i\sqrt{2}\alpha y^3}} {\sqrt{\pi 2^{n} n!}} \sum_{j=0}^n \binom{n}{j} H_j(-iy^3) (2y)^{n-j} \mathrm{Ai}^{(n-j)}(w)|_{w = y^4 - i\sqrt{2}\alpha y},
\end{eqnarray}
where $\mathrm{Ai}^{(n-j)}(w)$ is $(n-j)^{th}$ derivative of Airy function $\mathrm{Ai}(w)$. From \cite[Eq. 128]{chabaud2024complexity},
\begin{eqnarray}
    \mathrm{Ai}^{(\ell)}(w) &=& \mathrm{Ai}(0)f^{(\ell)}(w) +  \mathrm{Ai}'(0)g^{(\ell)}(w) \nonumber \\
&=& \frac{ \mathrm{Ai}(0)}{\Gamma(1/3)} \sum_{n=\lceil \ell/3 \rceil}^\infty \frac{3^n \Gamma(n+1/3)}{(3n-\ell)!} w^{3n-\ell}
+ \frac{ \mathrm{Ai}'(0)}{\Gamma(2/3)} \sum_{n=\lceil \ell/3 \rceil - 1}^\infty \frac{3^n \Gamma(n+2/3)}{(3n+1-\ell)!} w^{3n+1-\ell},
\end{eqnarray}
where $\mathrm{Ai}(0) = \frac{3^{-2/3}}{\Gamma(2/3)}$ and $\mathrm{Ai}'(0) = \frac{-3^{-1/3}}{\Gamma(1/3)}$. If $\tilde{\mathrm{Ai}}^{(l)}(w)$ is the approximated value of the $\ell$th derivative by taking first $d$ terms of the infinite series, then by \cite[Eq. 129]{chabaud2024complexity},
\begin{equation}
    |\mathrm{Ai}^{(l)}(w) - \tilde{\mathrm{Ai}}^{(l)}(w)| \leq \frac{1 + \frac{|w|}{2d}}{1 - \frac{3|w|^3}{2d}} \frac{(|w|)^{3d-l}}{(2d/3)^d}.
\end{equation}
Since $|w| = |y^4 - i\sqrt{2}\alpha y| = \exp$ for $\gamma,\alpha = \exp$. Therefore, taking $d = \exp$ with a sufficiently large polynomial in the exponential, we get
\begin{equation}
    |\mathrm{Ai}^{(l)}(w) - \tilde{\mathrm{Ai}}^{(l)}(w)| \leq 1/\eexp.
\end{equation}
To achieve this bound, we need to approximate $f^{(\ell)}(w)$ and $g^{(\ell)}(w)$ up to error $1/\eexp$ in parallel polynomial time.
We apply the tools developed in Appendix~\ref{sec:parallel}, computing products/divisions using \cite{BCH86}, $\Gamma(\cdot)$ using Lemma~\ref{lem:gamma}, and evaluating the polynomials in $w$ using Lemma~\ref{lem:fast-poly-eval}.
The magnitude of each of the terms in $f^{(\ell)}(w)$ (resp.\ $g^{(\ell)}(w)$) is $= \eexp$, therefore $f^{(\ell)}(w)$ (resp.\ $g^{(\ell)}(w)$) can be approximated up to error $1/\eexp$ in time $\exp$.

We then approximate $a_n(\alpha,\gamma)$ as
\begin{equation}
    \tilde{a}_n (\alpha,\gamma) = \frac{ye^{\frac23 y^6} e^{-\alpha_I^2-\alpha^2 - i\sqrt{2}\alpha y^3}} {\sqrt{\pi 2^{n} n!}} \sum_{j=0}^n \binom{n}{j} H_j(-iy^3) (2y)^{n-j} (\tilde{\mathrm{Ai}}^{(n-j)}(w))_{w = y^4 - i\sqrt{2}\alpha y}.
\end{equation}
This can be computed in parallel polynomial time with $n=\exp$ and precision $1/\eexp$, computing exponentials with Lemma~\ref{lem:exp}.
The error of the approximation is
\begin{eqnarray}
    |a_n (\alpha,\gamma) - \tilde{a}_n (\alpha,\gamma)| &\leq &  \frac{ye^{\frac23 y^6} e^{-\alpha_I^2-\alpha^2 - i\sqrt{2}\alpha y^3}} {\sqrt{\pi 2^{n} n!}} \nonumber \\ && \hspace{8mm} \times\sum_{j=0}^n \binom{n}{j} H_j(-iy^3) (2y)^{n-j} |(\mathrm{Ai}^{(n-j)}(w))_{w = y^4 - i\sqrt{2}\alpha y} - (\tilde{\mathrm{Ai}}^{(n-j)}(w))_{w = y^4 - i\sqrt{2}\alpha y}|. \nonumber \\
\end{eqnarray}
Since $|\binom{n}{j} H_j(-iy^3) (2y)^{n-j}| = \eexp$, $\tilde{a}_n (\alpha,\gamma)$ approximates $a_n(\alpha,\gamma)$ to precision $1/\eexp$ and can be computed in parallel polynomial time with $n=\exp$ and precision $1/\eexp$, computing summands in parallel, and the exponential function with Lemma~\ref{lem:exp}.
We have to calculate $N = \exp$ such Fock state coefficients, which we can also do in parallel. We define 
\begin{equation}
\ket{\psi_N} = \frac{1}{\sqrt{\mathcal{N}_N}} \sum_{n=0}^{N} \tilde{a}_n \ket{n},
\end{equation}
where the normalization factor is given by ${\mathcal{N}_N} = \sum_{n=0}^{\exp} |\tilde{a}_n|^2$ and can be computed in parallel polynomial time, since $N = \exp$. Now,
\begin{eqnarray}
    \frac12\|\ket{\psi_N'}\bra{\psi_N'} - \ket{\psi_N}\bra{\psi_N}\|_1 \leq \|\ket{\psi_N'} - \ket{\psi_N}\| &=& \sqrt{\sum_{n=0}^{N}\left|\frac{a_n - \tilde{a}_n}{\sqrt{\mathcal{N}_N}}\right|^2} + \frac{1}{\eexp} \nonumber \\ &\leq&  1/\eexp,
\end{eqnarray}
where $\|.\|_1$ refers to the trace distance, $\|.\|$ refers to the 2-norm.
Therefore,
\begin{eqnarray}
\frac12\|\ket{\psi_{\alpha,\gamma}}\bra{\psi_{\alpha,\gamma}} - \ket{\psi_N}\bra{\psi_N}\|_1 &\leq& \frac12\|\ket{\psi_{\alpha,\gamma}}\bra{\psi_{\alpha,\gamma}} - \ket{\psi_N'}\bra{\psi_N'}\|_1 +  \frac12\|\ket{\psi_N'}\bra{\psi_N'} - \ket{\psi_N}\bra{\psi_N}\|_1 \nonumber \\ &\leq& 1/\exp
\end{eqnarray}

\noindent Further, since for pure states,
\begin{equation}
    D(\ket{\psi_N}\bra{\psi_N},\ket{\psi_{\alpha,\gamma}}\bra{\psi_{\alpha,\gamma}}) = \sqrt{1 - F(\ket{\psi_N},\ket{\psi_{\alpha,\gamma}})} = \sqrt{1 - |\braket{\psi_N|\psi_{\alpha,\gamma}}|^2}.
\end{equation}
Therefore,
\begin{eqnarray}\label{eqn:lem_1_1}
    1 - |\braket{\psi_N|\psi_{\alpha,\gamma}}|^2 &\leq& 1/\exp, \nonumber \\
    |\braket{\psi_N|\psi_{\alpha,\gamma}}|^2  &\geq& 1 - 1/\exp.
\end{eqnarray}
Since
\begin{equation}
    \braket{\psi_N|\psi_{\alpha,\gamma}} = \frac{\sum_{i=0}^{N} \tilde{a}_n^* a_n}{\sqrt{\sum_{i=0}^{N} |\tilde{a}_n|^2}} = \sqrt{\sum_{i=0}^{N} |\tilde{a}_n|^2} + \frac{1}{\eexp}
\end{equation}
This is real and positive. Therefore, Eq.~(\ref{eqn:lem_1_1}) implies
\begin{equation}
    \Re(\braket{\psi_N|\psi_{\alpha,\gamma}}) = |\braket{\psi_N|\psi_{\alpha,\gamma}}| + \frac{1}{\eexp} \geq 1 - 1/\exp,
\end{equation}
up to precision $1/\eexp$. And finally
\begin{equation}
    \|\ket{\psi_{\alpha,\gamma}} - \ket{\psi_N}\| = \sqrt{2(1 - \Re(\braket{\psi_N|\psi_{\alpha,\gamma}}))} \leq 1/\exp.
\end{equation}
Now, having
\begin{equation}
    \ket{\psi_N} = \frac{1}{\sqrt{\mathcal{N}_N}} \sum_{i=0}^{\exp} \tilde{a}_n \ket{n} 
\end{equation}
we get the coherent state decomposition using \cite[Theorem 1]{Marshall2023},
\begin{eqnarray}
    \ket{\tilde{\psi}_{\alpha,\gamma}} &=& \frac{1}{\sqrt{\mathcal{N}}} \sum_{k=0}^{N} \tilde{c}_k \ket{\tilde{\epsilon} e^{2\pi ik/(N+1)}}, \\
    \tilde{c}_k &=& \frac{e^{\tilde{\epsilon}^2/2}}{(N+1)\sqrt{\mathcal{N}_N}} \sum_{n=0}^N \sqrt{n!} \frac{\tilde{a}_n}{\tilde{\epsilon}^n} e^{-2\pi ink/(N+1)},
\end{eqnarray}
where the normalization factor $\mathcal{N} = 1 + \mathcal{O}(\tilde{\epsilon}^{2(N+1)}/(N+1)!)$, which is $\delta$ close in trace distance to $\ket{\psi_N}$, with $\delta = \mathcal{O}(\tilde{\epsilon}^{(N+1)}/\sqrt{(N+1)!})$.
    
Since $N = \exp$, taking $\tilde{\epsilon} = \calO(1)$ suffices to get $\delta = 1/\exp$. Given $\tilde{a}_n$, $\forall n \in \{0,\dots,N\}$. To get $\tilde{c}_k$, we have to sum up $N=\exp$ terms, which takes parallel polynomial time.
We store an exponential number of coherent state coefficients, each with precision $1/\eexp$ (i.e., $\exp$ bits), which we compute in parallel.
In order to achieve a $2$-norm distance of $1/\exp$, it suffices to store the displacement parameters with a precision of $\poly$ bits, as $\|\ket{\alpha} - \ket{\beta}\| =  O(|\alpha-\beta|)$.
By construction,
\begin{equation}
    D(\ket{\psi_N}\bra{\psi_N},\ket{\tilde{\psi}_{\alpha,\gamma}}\bra{\tilde{\psi}_{\alpha,\gamma}}) \leq \delta = 1/\exp.
\end{equation}
This implies
\begin{equation}
    |\braket{\psi_N|\tilde{\psi}_{\alpha,\gamma}}|^2  \geq 1 - 1/\exp.
\end{equation}
Now,
\begin{equation}
    \braket{\psi_N|\tilde{\psi}_{\alpha,\gamma}} = \frac{1}{\sqrt{\mathcal{N}\mathcal{N}_N}} \sum_{n=0}^N \sum_{k=0}^N \tilde{a}_n^* \tilde{c}_k \braket{n|\tilde{\epsilon} e^{2\pi ik/(N+1)}}
 \end{equation}
Since we have the exact expression for the overlap of a coherent state with a given Fock state, $\braket{\psi_N|\tilde{\psi}_{\alpha,\gamma}}$ can calculated in time $\mathcal{O}(N^2) = \exp$. Lets say
\begin{equation}
    \braket{\psi_N|\tilde{\psi}_{\alpha,\gamma}} = re^{i\theta}
\end{equation}
by adding a global phase $\ket{\tilde{\psi}_{\alpha,\gamma}} \rightarrow e^{-i\theta} \ket{\tilde{\psi}_{\alpha,\gamma}}$, we get $\braket{\psi_N|\tilde{\psi}_{\alpha,\gamma}} = r$ real and positive. This implies
\begin{equation}
    \Re(\braket{\psi_N|\tilde{\psi}_{\alpha,\gamma}}) = |\braket{\psi_N|\tilde{\psi}_{\alpha,\gamma}}| \geq  1 - 1/\exp,
\end{equation}
and
\begin{equation}
    \|\ket{\tilde{\psi}_{\alpha,\gamma}} - \ket{\psi_N}\| = \sqrt{2(1 - \Re(\braket{\psi_N|\tilde{\psi}_{\alpha,\gamma}}))} \leq 1/\exp.
\end{equation}
Putting everything together, given $\ket{\psi_{\alpha,\gamma}}$, we obtain $\ket{\tilde{\psi}_{\alpha,\gamma}}$ in parallel polynomial time, such that
\begin{equation}
\|\ket{\psi_{\alpha,\gamma}} - \ket{\tilde{\psi}_{\alpha,\gamma}}\| \leq \|\ket{\psi_{\alpha,\gamma}} - \ket{\psi_N}\| + \|\ket{\psi_N} - \ket{\tilde{\psi}_{\alpha,\gamma}}\| \leq 1/\exp.
\end{equation}
\end{proof}
\section{Proof of Lemma \ref{lem:csr}}\label{appendixsec:theo_csr}
\noindent We restate Lemma \ref{lem:csr}:
\lemcsr*
\begin{proof}
    To calculate the trace distance between $\ket{\psi_\mathrm{out}}$ and $\ket{\tilde{\psi}_\mathrm{out}}$, and the time required to give the description of $\ket{\tilde{\psi}_\mathrm{out}}$, consider this: each time we encounter a cubic phase gate, we have a superposition of $N$ coherent states as the input, with $N \leq \exp$ (by Algorithm \ref{algo:decomp}, we know that each layer of cubic phase gate increases the number of terms in the superposition by a multiplicative factor $\exp$, we have $t$ such layers, so the maximum number of terms in the superposition at the input of any given cubic phase gate is $\exp^t = \exp$, for $t = \poly$). By doing a coherent state approximation of each of the $N$ states obtained by applying the cubic phase gate to the $N$ coherent states using Lemma \ref{lem:csr_decomp_cubic}, we obtain an output superposition of $N * \exp = \exp$ coherent states which can be described in parallel time $\poly\log N*\poly=\poly$ (adding the time to describe the coherent state approximation of each of the N terms with Lemma \ref{lem:csr_decomp_cubic}). To obtain the trace distance between the input and output superpositions of each cubic phase gate, we use the following technical result \cite{vomEnde2025}:
    \begin{lem}\label{lem:superpoSrace_distance}
         Given two normalized states $\ket{\Psi} = \sum_{i=1}^{n} c_i \ket{\psi_i}$ and $\ket{\Phi} = \sum_{i=1}^{n} c_i \ket{\phi_i}$, such that the state norm $\|\ket{\psi_i} - \ket{\phi_i}\| \leq \epsilon$, $\forall i \in \{1,\dots,n\}$, then the trace distance between $\ket{\Psi}$ and $\ket{\Phi}$ satisfies $D(\ket{\Psi}\bra{\Psi},\ket{\Phi}\bra{\Phi}) \leq \epsilon\sqrt{n}$.
    \end{lem}
    \noindent The proof of Lemma \ref{lem:superpoSrace_distance} is given at the end of the section. In our case, $n = \exp$, $\ket{\psi_i}$ is the state obtained by applying the cubic phase gate to the input coherent state and $\ket{\phi_i}$ refers to the approximation of this operation. Given $\epsilon = 1/\exp$ from Lemma \ref{lem:csr_decomp_cubic} and $n = \exp$ we conclude the trace distance distance between input and output superposition of coherent state for each cubic phase gate is $1/\exp$.
    
    We are making an approximation each time we encounter a cubic phase gate in the circuit and the approximated state moves away from the actual state by trace distance $1/\exp$. Since the effect of the displaced passive linear unitaries and the final Gaussian on the coherent states can be exactly computed, and unitaries do not increase the trace distance, then by triangle inequality the errors made at each of the $t$ layers add up and for $t = \poly$, $\ket{\tilde{\psi}_\mathrm{out}}$ is $t*1/\exp = 1/\exp$ close in trace distance to the actual output state $\ket{\psi_\mathrm{out}}$. The time required to give the description of $\ket{\tilde{\psi}_\mathrm{out}}$ is the time required to describe the coherent state decomposition for the $t$ cubic phase gates and the time required to describe the effect of the Gaussian $\hat{G}$ on the coherent states.
    For one such state, this can be done using $\mathcal{O}(m^3)$ arithmetic operations and $\calO(1)$ exponentials, according to \cite[Theorem III.7]{Dias2024}.
    As discussed in the proof of Lemma~\ref{lem:csr_decomp_cubic}, a precision of $1/\exp$ suffices, and we can perform the computation in parallel polynomial time as discussed in Appendix~\ref{sec:parallel}.
    This gives the total time to describe $\ket{\tilde{\psi}_\mathrm{out}}$ in parallel to be of the order $\mathcal{O}(t*\poly+m^3\poly) = \poly$.
\end{proof}

\begin{proof}[Proof of Lemma \ref{lem:superpoSrace_distance}]
Combining the triangle inequality and the Cauchy--Schwarz inequality, we show:
\begin{eqnarray}
\|\ket{\Psi} - \ket{\Phi}\| = \left\| \sum_{i=1}^n c_i (\ket{\psi_i} - \ket{\phi_i}) \right\| 
&\leq& \sum_{i=1}^n |c_i| \|\psi_i - \phi_i\| 
\leq \left( \sum_{i=1}^n |c_i|^2 \right)^{1/2} \left( \sum_{i=1}^n \|\ket{\psi_i} - \ket{\phi_i}\|^2 \right)^{1/2} \nonumber \\
&\leq& \left( \sum_{i=1}^n \epsilon^2 \right)^{1/2} = (n\epsilon^2)^{1/2} = \epsilon \sqrt{n}.
\end{eqnarray}
Then,
\begin{eqnarray}
\frac12 \| \ket{\Psi}\bra{\Psi} - \ket{\Phi}\bra{\Phi} \|_1 &=& \frac12\| \ket{\Psi}\bra{\Psi} - \ket{\Psi}\bra{\Phi} + \ket{\Psi}\bra{\Phi} - \ket{\Phi}\bra{\Phi} \|_1 \nonumber \\
&\leq&\frac12 \| \ket{\Psi}\bra{\Psi} - \ket{\Psi}\bra{\Phi} \|_1 + \frac12 \| \ket{\Psi}\bra{\Phi} - \ket{\Phi}\bra{\Phi} \|_1 \nonumber \\
&=& \frac12\| \ket{\Psi}\bra{\Psi - \Phi} \|_1 + \frac12 \| \ket{\Psi - \Phi}\bra{\Phi} \|_1 \nonumber \\
&=& \|\ket{\Psi} - \ket{\Phi}\| \leq \epsilon \sqrt{n}.
\end{eqnarray}
\end{proof}

\section{Proof of Lemma \ref{lem:prob}}\label{appendixsec:lem_prob}
\noindent We restate Lemma \ref{lem:prob}:
\lemprob*
\begin{proof}
Let $\ket{\tilde{\psi}_\mathrm{out},\tilde{n}}$ be the $m-1$ (possibly) unnormalized state obtained by projecting $\ket{\tilde{n}}$ onto $\ket{\tilde{\psi}_\mathrm{out}}$. Then,
    \begin{equation}
        \ket{\tilde{\psi}_\mathrm{out},\tilde{n}} = \braket{\tilde{n}|\tilde{\psi}_\mathrm{out}} = \sum_{j=0}^{n}\sum_{i=0}^{\exp} d_j^* c_i \braket{\beta_j|G_i} = \braket{\tilde{n}|\tilde{\psi}_\mathrm{out}} = \sum_{j=0}^{n}\sum_{i=0}^{\exp} d_j^* c_i * p(\beta_j) \braket{\beta_j|G_i}/p(\beta_j),
    \end{equation}
where $p(\beta_j) = \|\braket{\beta_j|G_i}\|^2$ can be calculated using $\mathcal{O}(1)$ operations \cite[Section III.D]{Dias2024}. From \cite[Theorem III.8]{Dias2024}, $\braket{\beta_j|G_i}/p(\beta_j)$ is an $(m-1)$-mode Gaussian state whose description can be computed using $\mathcal{O}(m^3)$ operations, of which $\calO(1)$ are exponentials.
We use the tools from Appendix~\ref{sec:parallel} to compute these operations in parallel polynomial time.
We have $n*\exp = \exp$ (for $n = \poly$) such total terms whose description needs to be computed, which can be done in parallel.
So the description of $\ket{\tilde{\psi}_\mathrm{out},\tilde{n}}$ can be given in time parallel polynomial time time.

Once we have $\ket{\tilde{\psi}_\mathrm{out},\tilde{n}}$, $\tilde{P}(n) = \|\ket{\tilde{\psi}_\mathrm{out},\tilde{n}}\|^2$ can be obtained by computing $n^2*\exp^2 = \exp$ Gaussian overlaps with $n = \poly$ and then summing them. Each overlap can be computed with $\mathcal{O}(m^3)$ operations according to \cite[Theorem III.5]{Dias2024}. Hence, we can approximate $\tilde{P}(n)$ in parallel polynomial time up to precision $1/\exp$. 
\end{proof}


\section{Proof of Lemma \ref{lem:csr_decomp_cubic_weak}}\label{appendixsec:csr_decomp_cubic_weak}
\noindent We first restate Lemma \ref{lem:csr_decomp_cubic_weak}:
\lemcsrdecompcubicweak*
\begin{proof}
   We first use the gentle measurement lemma (Lemma \ref{lem:gml}) to truncate $\ket{\psi_{\alpha,\gamma}}$ to a state with finite support size in the Fock state basis:
   \begin{equation}
       \ket{\psi_N'} = \frac{1}{\sqrt{\mathcal{N}_N'}}\sum_{i=0}^N a_n \ket{n},
   \end{equation}
\noindent with $a_n = \bra{n}e^{i\gamma \hat{q}^3}\ket{\alpha}$, $N = 4E/\epsilon^2$ and $\mathcal{N}_N' = \sum_{i=0}^N |a_n|^2$ such that
\begin{equation}
   \frac12 \|\ket{\psi_{\alpha,\gamma}}\bra{\psi_{\alpha,\gamma}} - \ket{\psi_N'} \bra{\psi_N'}\| \leq \frac{\epsilon}{2}.
\end{equation}
For $\ket{\psi_{\alpha,\gamma}} = e^{i\gamma\hat{q}^3}\ket{\alpha}$, $E$ is given by Eq.~(\ref{eqn:energy}) which is polynomial in $\alpha$ and $\gamma$. With \cite[Theorem 1]{Marshall2023} this can be approximated by a coherent state superposition with $N+1$ terms $\ket{\psi_{\alpha,\gamma}'}$, which is $\delta$-close in trace distance to $\ket{\psi_N'}$, where $\delta$ is a free parameter. Choosing $\delta = \epsilon/2$ ensures that this coherent state superposition is $\epsilon$-close in trace distance to $\ket{\psi_{\alpha,\gamma}}$, i.e.,
\begin{equation}
    \frac{1}{2} \|\ket{\psi_{\alpha,\gamma}} - \ket{\psi_{\alpha,\gamma}'}\|_1 \leq \frac{1}{2} \|\ket{\psi_{\alpha,\gamma}} - \ket{\psi_N'}\|_1 + \frac{1}{2} \|\ket{\psi_N'} - \ket{\psi_{\alpha,\gamma}'}\|_1 \leq \frac{\epsilon}{2} + \frac{\epsilon}{2} = \epsilon.
\end{equation}
Therefore, $\ket{\psi_{\alpha,\gamma}}$ has an $\epsilon$-approximate coherent state rank of at most $4E(\alpha,\gamma)/\epsilon^2 + 1$, with $E(\alpha,\gamma)$ given by Eq.~(\ref{eqn:energy}).
\end{proof}

\section{Parallel computation}\label{sec:parallel}

In this paper, we use the characterization of $\PSPACE$ as \emph{parallel polynomial time}, i.e., the class of decision problems that can be solved in $\poly(n)$ time using $\exp(n)$ parallel processors.
Let $\PTIME(t(n))$ be class of problems that can be solved in time $\calO(t(n))$ on a CREW\footnote{For the purpose of this paper, the memory model does not matter since a $p$-processor CRCW (concurrent read concurrent write) algorithm with runtime $t$ can be simulated by a $p$-processor EREW (exclusive read exclusive write) algorithm with runtime $\calO(t\log p)$.} PRAM (concurrent read exclusive write parallel random-access machine \cite{FW78,And86}), and $\DSPACE(s(n))$ the class of problems that can be solved in $\calO(s(n))$ space on a multitape DTM.
Then $\PTIME(t(n)) \subseteq \DSPACE(t^2(n))$ and $\DSPACE(s(n)) \subseteq \PTIME(s(n))$ for $s(n)\ge \log n$ \cite{FW78,And86}.
Thus, $\PTIME(\poly) = \PSPACE$.

An alternative characterization is given by $\NC$ \cite{Ruz81}, which is defined as the class of problems that can be decided by a $\log$-space uniform family of $\poly\log$-depth $\poly$-size Boolean circuits.
The scaled-up variant is then $\NCpoly$, with $\poly$-space uniform circuits of $\poly$-depth.
It holds that $\NCpoly=\PSPACE$ \cite{Bor77}.
Regarding $\NC$ as a function class, one can observe that $\NC$ and $\NCpoly$ compose well, i.e., if $f\in \NCpoly$ and $g\in \NC$, then $g\circ f\in \NCpoly$, which was also used in the $\mathsf{QIP}=\PSPACE$ proof \cite{JJUW10}.

We will require the fact that arithmetic (i.e., addition, multiplication, inverse, roots) can be performed in parallel polylogarithmic time in the bit length \cite{BCH86}.
We can then also compute the product or sum of $m$ $L$-bit values (e.g., $m!$) in time $\poly\log (m,L)$ using a tree-like computation structure.
Therefore, we can also evaluate polynomials with rational inputs efficiently:

\begin{lem}\label{lem:fast-poly-eval}
    Let $f(x) = \sum_{k=0}^d a_k x^k$ be a polynomial of degree $d$ with $a_0,\dots,a_d\in \QQ[i]$ given as inputs with numerators and denominators of at most $L$ bits.
    We can compute $f(x)$ with $x\in\QQ[i]$ of $L$ bits in parallel using $\poly(d,L)$ processors in parallel time $\poly\log(d,L)$.
\end{lem}

In our computations, we work with approximate values, truncated to an exponential number of bits in the input size (i.e., doubly exponentially small error).
Since $d$ will be at most exponential in the input size, we achieve sufficient accuracy while keeping the parallel runtime polynomial.

Our computations also compute roots, and inverses, which can be reduced to polynomial root isolation, which is in $\NC$ \cite{Nef94}.
As it turns out, even solutions to polynomial systems with $\calO(1)$ variables, including inequations and quantifiers, can be approximated in $\NC$ with exponential precision \cite{Ren92-4}.
Finally, we show how to approximate the exponential function and the Gamma function:

\begin{lem}\label{lem:exp}
    Let $z \in \CC$ be given as binary using $L$ bits.
    We can compute $y$ with $|y-e^z| \le \epsilon$ using $(L+|\log\epsilon|+|z|)^{\calO(1)}$ processors in time $\log^{\calO(1)}(L+|\log\epsilon|+|z|)$.
\end{lem}
\begin{proof}
    We approximate $e^z$ by truncating the Taylor series.
    Let $S_N(z) = \sum_{k=0}^{N-1}\frac{z^k}{k!}$ and $R_N(z) = \sum_{k=N}^{\infty}\frac{z^k}{k!}$.
    We want to choose $N$, such that $|R_N(z)| \le \epsilon$. For $N \ge 2e|z|$, we have
    \begin{equation}
        |R_N(z)| \le \sum_{k=N}^\infty \frac{|z|^k}{k!} \le \frac{|z|^N}{N!}\sum_{k=0}^\infty\left(\frac{|z|}{N+1}\right)^k  \le 2\frac{|z|^N}{N!}\le 2\frac{|z|^N}{\sqrt{2\pi N}(N/e)^N}\le 2^{-N+1}.
    \end{equation}
    Thus $N = \calO(|\log \epsilon|+|z|)$ suffices.
    We can compute $S_N(z)$ with Lemma~\ref{lem:fast-poly-eval} to achieve the desired parallel runtime, noting that $N!$ has $N^{\calO(1)}$-bits.
\end{proof}

\begin{lem}\label{lem:gamma}
    We can compute $m$ bits of $\Gamma(n/24)$ in parallel time $(\log mn)^{\calO(1)}$.
\end{lem}
\begin{proof}
    Using the recurrence relation $\Gamma(z+1)=z\Gamma(z)$, we can write
    \begin{equation}
        \Gamma(n+1/3) = (n-2/3)(n-5/3)\cdots\Gamma(1/3).
    \end{equation}
    We can efficiently compute the rational part following the above discussions, it only remains to approximate $\Gamma(1/3)$.
    It is shown in \cite{borwein1992fast} that $\calO(\log m)$ operations suffice to compute $m$ digits of $\Gamma(1/3)$ using Gaussian Arithmetic-Geometric Mean Iteration (AGM) \cite{BB87}.

    We remark that a more general result may be obtained via Spouge's approximation \cite{Spo94}.
\end{proof}

\end{document}